\documentclass[letterpaper,10pt,twoside]{article}
\usepackage[utf8]{inputenc}
\usepackage[english]{babel}
\usepackage[T1]{fontenc}
\usepackage{lmodern}
\usepackage{amsmath,amssymb,amsthm} 
\usepackage[ps,arrow,matrix]{xy} 

\newcommand{\im}{\mathrm{i}}
\newcommand{\N}{\mathbb{N}}
\newcommand{\R}{\mathbb{R}}
\newcommand{\C}{\mathbb{C}}
\newcommand{\defeq}{:=}
\newcommand{\tens}{\otimes}
\newcommand{\ctens}{\hat{\otimes}}
\newcommand{\ds}{\circ}
\newcommand{\cH}{\mathcal{H}}

\newcommand{\xd}{\mathrm{d}}

\newcommand{\toi}{\hookrightarrow}

\newcommand{\id}{\mathrm{id}}
\newcommand{\hr}{\mathrm{H}}
\newcommand{\sr}{\mathrm{S}}

\newcommand{\src}{\mathrm{S,c}}

\newcommand{\cohn}{\tilde{K}}

\newcommand{\cohs}{K^{\sr}}

\newcommand{\cohns}{\tilde{K}^{\sr}}

\newcommand{\coha}{\hat{K}}

\newcommand{\cohas}{\hat{K}^{\sr}}

\newcommand{\bsta}{\hat{\mathcal{B}}}

\newcommand{\Sq}{S^\mathrm{q}}
\newcommand{\obs}{\mathcal{O}}
\newcommand{\cobs}{\mathcal{C}}
\newcommand{\cobsl}{\mathcal{C}^{\mathrm{lin}}}
\newcommand{\cobsw}{\mathcal{C}^{\mathrm{Weyl}}}
\newcommand{\tord}{\mathbf{T}}
\newcommand{\aglue}{\diamond}
\newcommand{\quant}{\mathcal{Q}}
\newcommand{\one}{\mathbf{1}}
\newcommand{\no}[1]{:#1:}
\newcommand{\ano}[1]{\blacktriangleleft #1 \blacktriangleright}

\theoremstyle{definition}
\newtheorem{dfn}{Definition}[section]

\theoremstyle{plain}
\newtheorem{lem}[dfn]{Lemma}
\newtheorem{prop}[dfn]{Proposition}

\allowdisplaybreaks

\begin{document}

\begin{titlepage}
\title{\textbf{Schrödinger-Feynman quantization\\ and composition of observables in\\ general boundary quantum field theory}}
\author{Robert Oeckl\footnote{email: robert@matmor.unam.mx}\\ \\
Centro de Ciencias Matemáticas,\\
Universidad Nacional Autónoma de México,\\
Campus Morelia, C.P.~58190, Morelia, Michoacán, Mexico}
\date{UNAM-CCM-2012-1\\ 9 January 2012\\ 5 March 2015 (v2)}

\maketitle

\vspace{\stretch{1}}

\begin{abstract}
We show that the Feynman path integral together with the Schrödinger representation gives rise to a rigorous and functorial quantization scheme for linear and affine field theories. Since our target framework is the general boundary formulation, the class of field theories that can be quantized in this way includes theories without a metric spacetime background. We also show that this quantization scheme is equivalent to a holomorphic quantization scheme proposed earlier and based on geometric quantization.
We proceed to include observables into the scheme, quantized also through the path integral. We show that the quantized observables satisfy the canonical commutation relations, a feature shared with other quantization schemes also discussed. However, in contrast to other schemes the presented quantization also satisfies a correspondence between the composition of classical observables through their product and the composition of their quantized counterparts through spacetime gluing. In the special case of quantum field theory in Minkowski space this reproduces the operationally correct composition of observables encoded in the time-ordered product. We show that the quantization scheme also generalizes other features of quantum field theory such as the generating function of the S-matrix.

\end{abstract}

\vspace{\stretch{1}}
\end{titlepage}

\tableofcontents


\section{Introduction}

Quantum field theory is our most successful framework for describing the phenomena of the physical world at the subatomic scale. In particular, any quantum theory used to explain physical phenomena is thought of as derivable from (if not identical to) some quantum field theory. On the other hand, our conceptual framework of what fundamentally constitutes a quantum theory dates from our understanding of non-relativistic quantum phenomena, before the rise of quantum field theory. There are many quantum theories that fit into that framework, but that are not fundamental. On the other hand, quantum field theory can be fit into that framework, too, but it does so with a certain tension. In particular, there are many features of quantum field theory that, while compatible with, appear peculiar from the point of view of that framework. This suggests to rethink what fundamentally constitutes a quantum theory and in doing so take very serious the lessons quantum field theory has taught us. In light of the persistent failure to reconcile general relativity with the traditional framework of quantum theory it is perhaps high time to do so.

Evidently, a background Minkowski spacetime is not one of the features that we propose should play a more fundamental role. Rather it is structural features of quantum field theory that are instrumental to its predictive power, but that are not part of the traditional conceptual framework of a quantum theory. In particular, these are features of the S-matrix (such as crossing symmetry), the particular concept of observable (with its time-ordered product), and spacetime locality features (as conveniently encoded in the Feynman path integral). There is an ongoing effort to abstract these features from their quantum field theoretic context and make them part of a novel foundational framework for quantum theory. This is partly the subject of the present paper. This development can be taken to start in the 1980s with works of E.~Witten, G.~Segal and others, leading to \emph{topological quantum field theory} (TQFT) \cite{Ati:tqft} and with it a whole new branch of algebraic topology. The relevant features of quantum field theory abstracted here are locality properties, in particular the properties of the Feynman path integral. While this has mostly turned into an area of pure mathematics, G.~Segal in particular has developed a version of TQFT as a basis for conformal field theory \cite{Seg:cftproc,Seg:cftdef} and also for 4-dimensional quantum field theory \cite{Seg:fklectures}.

A proposal to take this strand of developments as the starting point of a foundational approach to quantum theory was elaborated in \cite{Oe:gbqft} (after initial suggestions in \cite{Oe:catandclock,Oe:boundary}) under the name of \emph{general boundary formulation} (GBF) of quantum theory. Crucially, the relevant ingredients from TQFT are complemented in this approach by a generalization of the Born rule for the extraction of probabilities for measurement processes. This makes possible a consistent probability interpretation which needs no reference to a metric spacetime background. In spite of this abstraction from metric spacetime, relevant features of quantum field theory are realized as fundamental properties in the GBF precisely in the spirit of our initial comments. In particular, the GBF inherits from TQFT the particular spacetime locality properties of amplitudes in quantum field theory usually encoded with the help of the Feynman path integral. Another feature of quantum theory that attains a more fundamental status in the GBF is crossing symmetry. This was indeed among the original motivations for introducing the GBF \cite{Oe:catandclock} and plays a crucial role in generalizations of the S-matrix based on the GBF \cite{CoOe:spsmatrix,CoOe:smatrixgbf,Col:desitterletter,Col:desitterpaper,CoDo:smatrixcsp}.

In light of the previous remarks the Feynman path integral suggests itself as a tool in the quantization of field theories in the GBF. This is naturally combined with the Schrödinger representation \cite{Oe:boundary,Oe:gbqft}. We shall refer to the quantization scheme determined by this combination as Schrödinger-Feynman quantization. This scheme has been extensively used to quantize field theories in the GBF. However, it has lacked so far a mathematically rigorous footing. To partially remedy this is one of the purposes of the present paper. Obviously, this can only be successful if we restrict to a sufficiently simple class of field theories. To this end we consider here linear field theory and the slightly more general affine field theory.

We start with a brief review of the geometric (Section~\ref{sec:geomax}) and algebraic (Section~\ref{sec:coreaxioms}) core structures of the GBF in Section~\ref{sec:gbfrev}. The rigorous and functorial version of the Schrödinger-Feynman quantization scheme is the subject of Section~\ref{sec:sfaffine}. To this end we first recall in Section~\ref{sec:classaffine} an axiomatization of affine field theory as put forward in \cite{Oe:affine}. In Section~\ref{sec:squant} the Schrödinger quantization on hypersurfaces is carried out, by importing the relevant treatment from \cite{Oe:schroedhol}. The core of this part of the paper is Section~\ref{sec:feynquant} where the Feynman path integral is rigorously defined. The key result here is that the Schrödinger-Feynman quantization scheme defined in this way is equivalent in a precise sense to the holomorphic quantization scheme put forward in \cite{Oe:affine}. This ensures on the one hand the functoriality of the former and serves on the other hand as an a posteriori justification of an ad hoc ingredient in the latter. In Section~\ref{sec:linft} we discuss the special case of linear field theory which is thus equivalent to the holomorphic quantization scheme put forward in \cite{Oe:holomorphic}.

Special features of the concept of observable in quantum field theory and their incorporation into the GBF are the subject of the second part of this paper, consisting of Section~\ref{sec:obs}. A proper concept of quantum observable for the GBF has been elaborated only recently \cite{Oe:obsgbf}. It was already suggested there that observables of quantum field theory fit more naturally into this concept than into the traditional one of an operator on Hilbert space. Moreover, it was suggested there that a striking correspondence between the spacetime composition of classical observables and their quantized counterparts is a key feature of quantum field theory. This was termed there \emph{composition correspondence}.  It is reviewed and elaborated on here in Section~\ref{sec:obsqft}. In Section~\ref{sec:obsgbf} the concept of observable in the GBF and the associated notion of composition is recalled from \cite{Oe:obsgbf} and refined. A minimalistic notion of classical observable is formalized in Section~\ref{sec:cobsax}. In Section~\ref{sec:compcor} quantization axioms are formulated which formalize in particular the notion of composition correspondence.

The core ingredient of this part of the paper is provided in Section~\ref{sec:feynweyl} with the quantization formula for Weyl observables (i.e., observables that are exponentials of imaginary linear observables), derived from the Feynman path integral. It is here where results of the first part of the paper (in particular Section~\ref{sec:feynquant}) are crucially employed. In Section~\ref{sec:classlinobs} an axiomatization of linear field theory with linear observables is given. The main result of this part of the paper is Section~\ref{sec:gbqftlow}, where it is shown that the quantization of the so formalized classical linear field theory satisfies not only the core axioms, but also the observable axioms and the quantization axioms, including composition correspondence. Factorization properties of observable maps are derived in Section~\ref{sec:factorization}. These abstract and generalize the generating function of the S-matrix in quantum field theory. While the treatment is focused on Weyl observables up to this point, a much more general class of observables is made accessible in Section~\ref{sec:genobs}. In Section~\ref{sec:diracquant} the more conventional operator product is derived for infinitesimal regions and it is shown to satisfy the canonical commutation relations. A comparison with other quantization schemes is carried out in Section~\ref{sec:compquant}. Finally, in Section~\ref{sec:conclusions} some conclusions and a brief outlook are presented.


\section{General boundary quantum field theory}
\label{sec:gbfrev}

In this section we briefly recall the geometric setting as well as the core axioms of the GBF in the form given in \cite{Oe:holomorphic,Oe:affine}. We shall call a model satisfying the core axioms a \emph{general boundary quantum field theory} (GBQFT).

\subsection{Geometric data}
\label{sec:geomax}

We recall briefly how in the GBF the structure of spacetime is formalized in terms of a \emph{spacetime system}, involving abstract notions of \emph{regions} and \emph{hypersurfaces}. We follow here closely the presentation in \cite{Oe:affine}. For further discussion of the rationale behind this, see in particular \cite{Oe:gbqft}.

There is a fixed positive integer $d\in\N$, the \emph{dimension} of spacetime. We are given a collection of oriented topological manifolds of dimension $d$, possibly with boundary, that we call \emph{regions}. Furthermore, there is a collection of oriented topological manifolds without boundary of dimension $d-1$ that we call \emph{hypersurfaces}. All manifolds may only have finitely many connected components. When we want to emphasize explicitly that a given manifold is in one of those collections we also use the attribute \emph{admissible}. These collections satisfy the following requirements:
\begin{itemize}
\item Any connected component of a region or hypersurface is admissible.
\item Any finite disjoint union of regions or of hypersurfaces is admissible.
\item Any boundary of a region is an admissible hypersurface.
\item If $\Sigma$ is a hypersurface, then $\overline{\Sigma}$, denoting the same manifold with opposite orientation, is admissible.
\end{itemize}
It will turn out to be convenient to also introduce \emph{slice regions}.\footnote{In previous papers slice regions were called \emph{empty regions}.} A slice region is topologically simply a hypersurface, but thought of as an infinitesimally thin region. Concretely, the slice region associated with a hypersurface $\Sigma$ will be denoted by $\hat{\Sigma}$ and its boundary is defined to be the disjoint union $\partial \hat{\Sigma}=\Sigma\cup\overline{\Sigma}$. There is one slice region for each hypersurface (forgetting its orientation). When an explicit distinction is desirable we refer to the previously defined regions as \emph{regular regions}.

There is also a notion of \emph{gluing} of regions. Suppose we are given a region $M$ with its boundary a disjoint union $\partial M=\Sigma_1\cup\Sigma\cup\overline{\Sigma'}$, where $\Sigma'$ is a copy of $\Sigma$. ($\Sigma_1$ may be empty.) Then, we may obtain a new manifold $M_1$ by gluing $M$ to itself along $\Sigma,\overline{\Sigma'}$. That is, we identify the points of $\Sigma$ with corresponding points of $\Sigma'$ to obtain $M_1$. The resulting manifold $M_1$ might be inadmissible, in which case the gluing is not allowed.

Depending on the theory one wants to model, the manifolds may carry additional structure such as for example a differentiable structure or a metric. This has to be taken into account in the gluing and will modify the procedure as well as its possibility in the first place. Our description above is merely meant as a minimal one. Moreover, there might be important information present in different ways of identifying the boundary hypersurfaces that are glued. Such a case can be incorporated into our present setting by encoding this information explicitly through suitable additional structure on the manifolds.

For brevity we shall refer to a collection of regions and hypersurfaces with the properties given above as a \emph{spacetime system}. A spacetime system can be induced from a global spacetime manifold by taking suitable submanifolds. (This setting was termed a \emph{global background} in \cite{Oe:gbqft}.) On the other hand, a spacetime system may arise by considering regions as independent pieces of spacetime that are not a priori embedded into any global manifold. Indeed, depending on the context, it might be physically undesirable to assume knowledge of, or even existence of, a fixed global spacetime structure.

\subsection{Core axioms}
\label{sec:coreaxioms}

A GBQFT on a spacetime system is a model satisfying the following axioms.

\begin{itemize}
\item[(T1)] Associated to each hypersurface $\Sigma$ is a complex
  separable Hilbert space $\cH_\Sigma$, called the \emph{state space} of
  $\Sigma$. We denote its inner product by
  $\langle\cdot,\cdot\rangle_\Sigma$.
\item[(T1b)] Associated to each hypersurface $\Sigma$ is a conjugate linear
  isometry $\iota_\Sigma:\cH_\Sigma\to\cH_{\overline{\Sigma}}$. This map
  is an involution in the sense that $\iota_{\overline{\Sigma}}\circ\iota_\Sigma$
  is the identity on  $\cH_\Sigma$.
\item[(T2)] Suppose the hypersurface $\Sigma$ decomposes into a disjoint
  union of hypersurfaces $\Sigma=\Sigma_1\cup\cdots\cup\Sigma_n$. Then,
  there is an isometric isomorphism of Hilbert spaces
  $\tau_{\Sigma_1,\dots,\Sigma_n;\Sigma}:\cH_{\Sigma_1}\ctens\cdots\ctens\cH_{\Sigma_n}\to\cH_\Sigma$.
  The composition of the maps $\tau$ associated with two consecutive
  decompositions is identical to the map $\tau$ associated to the
  resulting decomposition.
\item[(T2b)] The involution $\iota$ is compatible with the above
  decomposition. That is,
  $\tau_{\overline{\Sigma}_1,\dots,\overline{\Sigma}_n;\overline{\Sigma}}
  \circ(\iota_{\Sigma_1}\ctens\cdots\ctens\iota_{\Sigma_n}) 
  =\iota_\Sigma\circ\tau_{\Sigma_1,\dots,\Sigma_n;\Sigma}$.
\item[(T4)] Associated with each region $M$ is a linear map
  from a dense subspace $\cH_{\partial M}^\ds$ of the state space
  $\cH_{\partial M}$ of its boundary $\partial M$ (which carries the
  induced orientation) to the complex
  numbers, $\rho_M:\cH_{\partial M}^\ds\to\C$. This is called the
  \emph{amplitude} map.
\item[(T3x)] Let $\Sigma$ be a hypersurface. The boundary $\partial\hat{\Sigma}$ of the associated slice region $\hat{\Sigma}$ decomposes into the disjoint union $\partial\hat{\Sigma}=\overline{\Sigma}\cup\Sigma'$, where $\Sigma'$ denotes a second copy of $\Sigma$. Then, $\tau_{\overline{\Sigma},\Sigma';\partial\hat{\Sigma}}(\cH_{\overline{\Sigma}}\tens\cH_{\Sigma'})\subseteq\cH_{\partial\hat{\Sigma}}^\ds$. Moreover, $\rho_{\hat{\Sigma}}\circ\tau_{\overline{\Sigma},\Sigma';\partial\hat{\Sigma}}$ restricts to a bilinear pairing $(\cdot,\cdot)_\Sigma:\cH_{\overline{\Sigma}}\times\cH_{\Sigma'}\to\C$ such that $\langle\cdot,\cdot\rangle_\Sigma=(\iota_\Sigma(\cdot),\cdot)_\Sigma$.
\item[(T5a)] Let $M_1$ and $M_2$ be regions and $M\defeq M_1\cup M_2$ be their disjoint union. Then $\partial M=\partial M_1\cup \partial M_2$ is also a disjoint union and $\tau_{\partial M_1,\partial M_2;\partial M}(\cH_{\partial M_1}^\ds\tens \cH_{\partial M_2}^\ds)\subseteq \cH_{\partial M}^\ds$. Moreover, for all $\psi_1\in\cH_{\partial M_1}^\ds$ and $\psi_2\in\cH_{\partial M_2}^\ds$,
\begin{equation}
 \rho_{M}\circ\tau_{\partial M_1,\partial M_2;\partial M}(\psi_1\tens\psi_2)= \rho_{M_1}(\psi_1)\rho_{M_2}(\psi_2) .
\label{eq:glueaxa}
\end{equation}
\item[(T5b)] Let $M$ be a region with its boundary decomposing as a disjoint union $\partial M=\Sigma_1\cup\Sigma\cup \overline{\Sigma'}$, where $\Sigma'$ is a copy of $\Sigma$. Let $M_1$ denote the gluing of $M$ with itself along $\Sigma,\overline{\Sigma'}$ and suppose that $M_1$ is a region. Note $\partial M_1=\Sigma_1$. Then, $\tau_{\Sigma_1,\Sigma,\overline{\Sigma'};\partial M}(\psi\tens\xi\tens\iota_\Sigma(\xi))\in\cH_{\partial M}^\ds$ for all $\psi\in\cH_{\partial M_1}^\ds$ and $\xi\in\cH_\Sigma$. Moreover, for any ON-basis $\{\xi_i\}_{i\in I}$ of $\cH_\Sigma$, we have for all $\psi\in\cH_{\partial M_1}^\ds$,
\begin{equation}
 \rho_{M_1}(\psi)\cdot c(M;\Sigma,\overline{\Sigma'})
 =\sum_{i\in I}\rho_M\circ\tau_{\Sigma_1,\Sigma,\overline{\Sigma'};\partial M}\left(\psi\tens\xi_i\tens\iota_\Sigma(\xi_i)\right),
\label{eq:glueaxb}
\end{equation}
where $c(M;\Sigma,\overline{\Sigma'})\in\C\setminus\{0\}$ is called the \emph{gluing anomaly factor} and depends only on the geometric data.
\end{itemize}

For ease of notation we will use the maps $\tau$ implicitly in the following, omitting their explicit mention.


\section{Schrödinger-Feynman quantization of affine field theory}
\label{sec:sfaffine}

Most quantum field theories of physical interest are at least in part based on the \emph{quantization} of a classical field theory. Such a quantization proceeds by transforming ingredients of the classical field theory into ingredients of a quantum field theory, following a more or less heuristic \emph{quantization scheme}.

In the case of the general boundary formulation (GBF) the objects that the quantization scheme has to produce are principally the Hilbert spaces associated to hypersurfaces and the amplitude maps associated to regions. Since these structures differ from those usually taken to define a quantum theory, most quantization schemes are at least not immediately adaptable to the GBF. An exception is the Schrödinger-Feynman quantization scheme. Here, the Hilbert space on each hypersurface is constructed according to the Schrödinger prescription, i.e., as a space of square-integrable functions on the space of field configurations on the hypersurface. The amplitude map for a region is constructed as the Feynman path integral over all field configurations in the region, evaluated with the boundary state inserted. It is fair to say that this quantization scheme was an essential ingredient in the motivations that lead to the emergence of the field of topological quantum field theory (TQFT) in the 1980s. Since the core axioms of the GBF may be seen as a specific variant of TQFT, it is unsurprising that Schrödinger-Feynman quantization is a natural candidate for a quantization scheme putting out quantum theories in GBF form \cite{Oe:boundary,Oe:gbqft}. Indeed, Schrödinger-Feynman quantization has been successfully implemented in the context of the GBF in various examples \cite{Oe:kgtl,Oe:2dqym,CoOe:spsmatrix,CoOe:smatrixgbf,CoOe:2deucl,Col:desitterletter,Col:desitterpaper}. However, even when leading to rigorous results in many cases, the quantization scheme is presented in those papers in a rather heuristic and non-rigorous form. This impedes its wider applicability, especially in the context of more complex geometric situations or for more complicated field theories.

We shall present in this section a fully rigorous and functorial version of Schrödinger-Feynman quantization for the GBF. Unsurprisingly, this comes at the cost of specialization on the field theory side. Indeed, we shall limit our considerations to affine field theory, i.e., where the spaces of local solutions of the field theory are naturally affine spaces. Linear field theory arises as the special case where local spaces of solutions have a special point and are thus linear spaces.

Fortunately, rather than having to construct the quantization scheme from scratch we can rely on the ``hard work'' done elsewhere, namely in the papers \cite{Oe:holomorphic,Oe:affine,Oe:schroedhol}.
Firstly, in \cite{Oe:schroedhol} a rigorous construction of the Schrödinger representation was given. This is precisely suitable to construct the Hilbert spaces associated to hypersurfaces. Secondly, we show that the heuristic Feynman path integral prescription leads to a precise definition of the amplitude map for regions. Thirdly, we recall the rigorous and functorial quantization scheme that was established using the holomorphic representation on hypersurfaces for linear field theory in \cite{Oe:holomorphic} and for affine field theory in \cite{Oe:affine}. Using further results of \cite{Oe:schroedhol} we bring the output of the two quantization schemes into a one-to-one correspondence. This allows us to conclude that the Schrödinger-Feynman quantization scheme indeed yields a GBQFT, i.e., satisfies the GBF core axioms as was proven for the holomorphic quantization scheme in \cite{Oe:holomorphic,Oe:affine}.

\subsection{Encoding classical affine field theory}
\label{sec:classaffine}

Recall that a set $A$ is an \emph{affine space} over a real vector space $L$ if there is a transitive and free abelian group action $L\times A\to A$, called \emph{translation}. As is customary, we shall write this action as addition, i.e, $(l,a)\mapsto l+a$ for $l\in L$ and $a\in A$. Also we shall be indiscriminate about the order, writing $l+a=a+l$. Given a \emph{base point} $e\in A$ we obtain a canonical identification of $L$ with $A$ via $l\mapsto l+e$.

An affine field theory is a field theory such that the local spaces of solutions are naturally affine spaces. We briefly recall from \cite{Oe:affine} how such field theories may be axiomatically formalized given a spacetime system. For each region $M$ we denote the affine space of solutions in $M$ by $A_M$ and the associated real vector space by $L_M$. Note that $L_M$ is canonically identified with the tangent space of $A_M$ at each point. Similarly, for a hypersurface $\Sigma$ we denote the space of (germs of) solutions in a neighborhood of $\Sigma$ by $A_\Sigma$. $L_\Sigma$ denotes the associated real vector space. Given a Lagrangian that induces the field theory this leads to further natural structures. For each region $M$ this is the action $S_M:A_M\to\R$. For each hypersurface $\Sigma$ this is the symplectic potential $\theta_\Sigma:A_\Sigma\times L_\Sigma\to\R$ and its exterior derivative, the symplectic form $\omega_\Sigma:L_\Sigma\times L_\Sigma\to\R$.

The following list of axioms from \cite{Oe:affine} is meant to capture precisely these ingredients of affine field theory and their interrelations in a way that is as universal as possible while being reasonably minimal. We remark that there is one further ingredient in the axioms below that, rather than being part of the classical field theory, already is a seed for its quantization. This is a complex structure $J_\Sigma$ on $L_\Sigma$ for each hypersurface $\Sigma$, which also partially determines the inner products $g_\Sigma$ and $\{\cdot,\cdot\}_\Sigma$. It will be discussed it in the following section.

\begin{itemize}
\item[(C1)] Associated to each hypersurface $\Sigma$ is a complex separable Hilbert space $L_\Sigma$ and an affine space $A_\Sigma$ over $L_\Sigma$ with the induced topology. The latter means that there is a transitive and free abelian group action $L_\Sigma\times A_\Sigma\to A_\Sigma$ which we denote by $(\phi,\eta)\mapsto \phi+\eta$. The inner product in $L_\Sigma$ is denoted by $\{\cdot ,\cdot\}_\Sigma$. We also define $g_\Sigma(\cdot,\cdot)\defeq \Re\{\cdot ,\cdot\}_\Sigma$ and $\omega_\Sigma(\cdot,\cdot)\defeq \frac{1}{2}\Im\{\cdot ,\cdot\}_\Sigma$ and denote by $J_\Sigma:L_\Sigma\to L_\Sigma$ the scalar multiplication with $\im$ in $L_\Sigma$. Moreover we suppose there are continuous maps $\theta_\Sigma:A_\Sigma\times L_\Sigma\to\R$ and $[\cdot,\cdot]_\Sigma:L_\Sigma\times L_\Sigma\to\R$ such that $\theta_\Sigma$ is real linear in the second argument, $[\cdot,\cdot]_\Sigma$ is real bilinear, and both structures are compatible via
\begin{equation}
[\phi,\phi']_\Sigma+\theta_\Sigma(\eta,\phi')=\theta_\Sigma(\phi+\eta,\phi')\qquad\forall \eta\in A_\Sigma, \forall \phi,\phi'\in L_\Sigma .
\end{equation}
Finally we require
\begin{equation}
 \omega_\Sigma(\phi,\phi')=\frac{1}{2} [\phi,\phi']_\Sigma-\frac{1}{2} [\phi',\phi]_\Sigma
\qquad \forall  \phi,\phi'\in L_\Sigma .
\end{equation}
\item[(C2)] Associated to each hypersurface $\Sigma$ there is a homeomorphic involution $A_\Sigma\to A_{\overline{\Sigma}}$ and a compatible conjugate linear involution $L_\Sigma\to L_{\overline\Sigma}$ under which the inner product is complex conjugated. We will not write these maps explicitly, but rather think of $A_\Sigma$ as identified with $A_{\overline{\Sigma}}$ and $L_\Sigma$ as identified with $L_{\overline\Sigma}$. Then, $\{\phi',\phi\}_{\overline{\Sigma}}=\overline{\{\phi',\phi\}_\Sigma}$ and we also require $\theta_{\overline{\Sigma}}(\eta,\phi)=-\theta_\Sigma(\eta,\phi)$ and $[\phi,\phi']_{\overline{\Sigma}}=-[\phi,\phi']_\Sigma$ for all $\phi,\phi'\in L_\Sigma$ and $\eta\in A_\Sigma$.
\item[(C3)] Suppose the hypersurface $\Sigma$ decomposes into a disjoint
  union of hypersurfaces $\Sigma=\Sigma_1\cup\cdots\cup\Sigma_n$. Then,
  there is a homeomorphism $A_{\Sigma_1}\times\dots\times A_{\Sigma_n}\to A_\Sigma$ and a compatible isometric isomorphism of complex Hilbert spaces
  $L_{\Sigma_1}\oplus\cdots\oplus L_{\Sigma_n}\to L_\Sigma$. Moreover, these maps satisfy obvious associativity conditions. We will not write these maps explicitly, but rather think of them as identifications. Also, $\theta_\Sigma=\theta_{\Sigma_1}+\dots+\theta_{\Sigma_n}$ and $[\cdot,\cdot]_\Sigma=[\cdot,\cdot]_{\Sigma_1}+\dots+[\cdot,\cdot]_{\Sigma_n}$.
\item[(C4)] Associated to each region $M$ is a real vector space $L_M$ and an affine space $A_M$ over $L_M$. Also, there is a map $S_M:A_M\to\R$.
\item[(C5)] Associated to each region $M$ there is a map $a_M:A_M\to A_{\partial M}$ and a compatible linear map of real vector spaces $r_M:L_M\to L_{\partial M}$. We denote by $A_{\tilde{M}}$ the image of $A_M$ under $a_M$ and by $L_{\tilde{M}}$ the image of $L_M$ under $r_M$. $L_{\tilde{M}}$ is a closed Lagrangian subspace of the real Hilbert space $L_{\partial M}$ with respect to the symplectic form $\omega_{\partial M}$. We often omit the explicit mention of the maps $a_M$ and $r_M$. We also require $S_M(\eta)=S_M(\eta')$ if $a_M(\eta)=a_M(\eta')$, and
\begin{equation}
S_M(\eta)=S_M(\eta')-\frac{1}{2}\theta_{\partial M}(\eta,\eta-\eta')-\frac{1}{2}\theta_{\partial M}(\eta',\eta-\eta')\qquad\forall\eta,\eta'\in A_M .
\label{eq:actsympot}
\end{equation}
\item[(C6)] Let $M_1$ and $M_2$ be regions and $M\defeq M_1\cup M_2$ be their disjoint union. Then, there is a bijection $A_{M_1}\times A_{M_2}\to A_M$ and a compatible isomorphism of real vector spaces $L_{M_1}\oplus L_{M_2}\to L_M$ such that $a_M=a_{M_1}\times a_{M_2}$ and $r_M=r_{M_1}\times r_{M_2}$. Moreover these maps satisfy obvious associativity conditions. Hence, we can think of them as identifications and omit their explicit mention in the following. We also require $S_M=S_{M_1}+S_{M_2}$.
\item[(C7)] Let $M$ be a region with its boundary decomposing as a disjoint union $\partial M=\Sigma_1\cup\Sigma\cup \overline{\Sigma'}$, where $\Sigma'$ is a copy of $\Sigma$. Let $M_1$ denote the gluing of $M$ to itself along $\Sigma,\overline{\Sigma'}$ and suppose that $M_1$ is a region. Note $\partial M_1=\Sigma_1$. Then, there is an injective map $a_{M;\Sigma,\overline{\Sigma'}}:A_{M_1}\toi A_{M}$ and a compatible injective linear map $r_{M;\Sigma,\overline{\Sigma'}}:L_{M_1}\toi L_{M}$ such that
\begin{equation}
 A_{M_1}\toi A_{M}\rightrightarrows A_\Sigma \qquad L_{M_1}\toi L_{M}\rightrightarrows L_\Sigma
\end{equation}
are exact sequences. Here, for the first sequence, the arrows on the right hand side are compositions of the map $a_M$ with the projections of $A_{\partial M}$ to $A_\Sigma$ and $A_{\overline{\Sigma'}}$ respectively (the latter identified with $A_\Sigma$). For the second sequence the arrows on the right hand side are compositions of the map $r_M$ with the projections of $L_{\partial M}$ to $L_\Sigma$ and $L_{\overline{\Sigma'}}$ respectively (the latter identified with $L_\Sigma$). 
We also require $S_{M_1}=S_M\circ a_{M;\Sigma,\overline{\Sigma'}}$.
Moreover, the following diagrams commute, where the bottom arrows are the projections.
\begin{equation}
\xymatrix{
  A_{M_1} \ar[rr]^{a_{M;\Sigma,\overline{\Sigma'}}} \ar[d]_{a_{M_1}} & & A_{M} \ar[d]^{a_{M}}\\
  A_{\partial M_1}  & & A_{\partial M} \ar[ll]} \qquad
\xymatrix{
  L_{M_1} \ar[rr]^{r_{M;\Sigma,\overline{\Sigma'}}} \ar[d]_{r_{M_1}} & & L_{M} \ar[d]^{r_{M}}\\
  L_{\partial M_1}  & & L_{\partial M} \ar[ll]}
\end{equation}
\end{itemize}

\subsection{Schrödinger quantization on hypersurfaces}
\label{sec:squant}

In this section we consider the first part of the Schrödinger-Feynman quantization scheme which consists in associating to each hypersurface $\Sigma$ a Hilbert space $\cH^\sr_{\Sigma}$ of states in accordance with the Schrödinger prescription. We differ here in our notation from that of the core axioms (Section~\ref{sec:coreaxioms}) to emphasize that $\cH^\sr_{\Sigma}$ arises from a particular quantization scheme. Nevertheless, it is understood that $\cH^\sr_{\Sigma}$ is taking the place of $\cH_{\Sigma}$ in the core axioms.

The basic idea of the Schrödinger prescription is to construct the Hilbert space $\cH^\sr_{\Sigma}$ as a space of square-integrable \emph{wave functions} on the configuration space associated with the hypersurface $\Sigma$. We shall denote this real affine space of configurations by $C_{\Sigma}$ in accordance with the notation used in \cite{Oe:schroedhol}. We recall from \cite{Oe:schroedhol} that $C_\Sigma$ can be obtained in a simple way from the space $A_{\Sigma}$ and the symplectic potential $\theta_{\Sigma}$. Concretely, define the subspaces $M_{\Sigma}\subseteq L_{\Sigma}$ and $N_{\Sigma}\subseteq L_{\Sigma}$ as
\begin{align}
 M_{\Sigma} & \defeq\{\tau\in L_{\Sigma}: [\xi,\tau]_{\Sigma}=0\; \forall \xi\in L_\Sigma\} \\
 N_{\Sigma} & \defeq\{\tau\in L_{\Sigma}: [\tau,\xi]_{\Sigma}=0\; \forall \xi\in L_{\Sigma}\} .
\end{align}
$M_{\Sigma}$ should be thought of as the subspace of momenta, i.e., those (infinitesimal) solutions where field values vanish on $\Sigma$ while their derivatives do not. For this to really make sense we need the additional requirement on the map $[\cdot,\cdot]_{\Sigma}$ that $M_{\Sigma}$ and $N_{\Sigma}$ together generate all of $L_{\Sigma}$ \cite{Oe:schroedhol}. This then even implies $L_{\Sigma}=M_{\Sigma}\oplus N_{\Sigma}$.

The configuration space $C_\Sigma$ is the quotient space
\begin{equation}
C_\Sigma\defeq A_\Sigma/M_\Sigma .
\end{equation}
For later use we denote the quotient map by $c_\Sigma:A_\Sigma\to C_\Sigma$.
The infinitesimal version of $C_M$ is the linear quotient space
\begin{equation}
Q_\Sigma\defeq L_\Sigma/M_\Sigma .
\end{equation}
$C_\Sigma$ is thus an affine space over $Q_\Sigma$. For later use we denote the quotient map by $q_\Sigma:L_\Sigma\to Q_\Sigma$.

In the case of linear field theory (i.e., when $C_\Sigma=Q_\Sigma$) it is well known that determining the Hilbert space of the Schrödinger representation requires the additional datum of a \emph{vacuum state}, see e.g., \cite{Jac:schroedinger}. It is convenient \cite{Oe:schroedhol} to encode this in a symmetric bilinear form $\Omega_\Sigma:Q_\Sigma\times Q_\Sigma\to\C$ with positive definite real part. The vacuum state is then determined by the wave function $\cohs_0:Q_\Sigma\to\C$ given by
\begin{equation}
\cohs_0(\phi)=\exp\left(-\frac{1}{2}\Omega(\phi,\phi)\right) .
\end{equation}
In the more general case of affine field theory there is no special vacuum state. Nevertheless, the Schrödinger representation is still determined by a symmetric bilinear form $\Omega_\Sigma:Q_\Sigma\times Q_\Sigma\to\C$ with positive definite real part \cite{Oe:schroedhol}.

In axiom (C1) we are given the affine space $A_\Sigma$, the associated linear space $L_\Sigma$ and the maps $\theta_\Sigma$, $[\cdot,\cdot]_\Sigma$ as well as $\omega_\Sigma$. In order to define the Schrödinger representation Hilbert space we are then only missing $\Omega_\Sigma$.\footnote{We are simplifying the discussion here by omitting another ingredient missing in axiom (C1). This is the condition $M_{\Sigma}+N_{\Sigma}=L_{\Sigma}$ necessary for the Schrödinger representation to be well defined. However, since the whole quantization scheme will turn out to be equivalent to a holomorphic quantization scheme which does not require this condition, we have omitted it in the axiom in the first place.} However, it was shown in \cite{Oe:schroedhol} that the complex structure $J_\Sigma$, also provided in axiom (C1), precisely gives rise to such a bilinear form $\Omega_\Sigma$. Concretely, define $j_\Sigma:Q_\Sigma\to L_\Sigma$ to be the unique linear map such that $q_\Sigma\circ j_\Sigma=\id_{Q_\Sigma}$ and $j_\Sigma(Q_\Sigma)=J_\Sigma M_\Sigma$. Then,
\begin{equation}
\Omega_\Sigma(\phi,\phi') = g_\Sigma(j_\Sigma(\phi),j_\Sigma(\phi'))
 -\im [j_\Sigma(\phi),j_\Sigma(\phi')]_\Sigma
\label{eq:omegaj}
\end{equation}
is a bilinear form precisely as required. What is more, in \cite{Oe:schroedhol} it was shown that admissible bilinear forms $\Omega_\Sigma$ are in one-to-one correspondence to admissible complex structures $J_\Sigma$ precisely via equation (\ref{eq:omegaj}).\footnote{While the exact definition of \emph{admissible complex structure} is implicit in axiom (C1), we repeat here the exact definition of \emph{admissible bilinear form}. $\Omega_\Sigma:Q_\Sigma\times Q_\Sigma\to\C$ is admissible if it is a symmetric bilinear form such that its real part is a positive definite inner product making $Q_\Sigma$ into a real Hilbert space and such that its imaginary part is continuous with respect to this Hilbert space structure.} Thus, the datum of the complex structure in axiom (C1) is precisely equivalent to the datum of an admissible bilinear form. Hence, the data of axiom (C1) uniquely determine a Schrödinger representation Hilbert space $\cH^\sr_\Sigma$.

The precise nature of the construction of the Hilbert space $\cH^\sr_\Sigma$ is not relevant here and we refer the interested reader to \cite{Oe:schroedhol} for the details. A fact about the Schrödinger representation that we will make use of, however, is that $\cH^\sr_\Sigma$ contains a dense subspace generated by coherent states. We recall from \cite{Oe:schroedhol} that the Schrödinger wave function of the affine coherent state \cite{Oe:affine} associated to the local solution $\zeta\in A_\Sigma$ is given by\footnote{While coherent states can be represented by wave functions in this way, this is not true of all states in the Hilbert space $\cH^\sr_\Sigma$. More generally, $\cH^\sr_\Sigma$ should be thought of as a space of \emph{reduced wave functions} on a linearized and extended version of $C_\Sigma$, see \cite{Oe:schroedhol}. However, as the only concrete states we need to consider are coherent states we may ignore these details for the purposes of the present paper.}
\begin{equation}
 \cohas_\zeta(\varphi)=\exp\left(\im\,\theta_\Sigma(\zeta,\varphi-c_\Sigma(\zeta))-\frac{1}{2}\Omega_\Sigma(\varphi-c_\Sigma(\zeta),\varphi-c_\Sigma(\zeta))\right) .
\label{eq:acohs}
\end{equation}

As already mentioned, we shall make use of the fact that the Schrödinger representation can be brought into correspondence with the holomorphic representation. We recall that the Hilbert space of the holomorphic representation is a space of square-integrable holomorphic functions on $L_\Sigma$. The holomorphic representation arises as a special case of geometric quantization \cite{Woo:geomquant} and depends precisely on an admissible complex structure $J_\Sigma$ on $L_\Sigma$ as exhibited in axiom (C1). Indeed, the ingredients of axiom (C1) precisely determine a Hilbert space $\cH^\hr_\Sigma$ of the holomorphic representation \cite{Oe:affine}. Now, it was shown in \cite{Oe:schroedhol} that the one-to-one correspondence (\ref{eq:omegaj}) between admissible bilinear forms $\Omega_\Sigma$ and admissible complex structures $J_\Sigma$ induces a canonical isometric isomorphism of Hilbert spaces $\bsta_\Sigma:\cH^\sr_\Sigma\to\cH^\hr_\Sigma$. This means that on the level of hypersurfaces and associated Hilbert spaces the present Schrödinger quantization is precisely equivalent to the holomorphic quantization exhibited in \cite{Oe:affine}. In particular, core axioms (T1), (T1b), (T2), and (T2b) are satisfied, as was shown in \cite{Oe:affine} for the holomorphic quantization.

\subsection{Feynman quantization in regions}
\label{sec:feynquant}

In this section we consider the second part of the Schrödinger-Feynman quantization scheme which consists in associating to each region $M$ an amplitude map $\rho^\sr_M:\cH^{\sr\ds}_{\partial M}\to\C$ from a dense subspace of the the boundary Hilbert space $\cH^{\sr}_{\partial M}$ to the complex numbers. We shall proceed by heuristically and non-rigorously following the Feynman path integral prescription. However, this yields a definite and well defined result which we then take as a definition.

If $M$ is a spacetime region and $\psi^\sr$ the wave function of a state in the Schrödinger representation Hilbert space $\cH^\sr_{\partial M}$, its amplitude is given heuristically by the Feynman path integral via
\begin{equation}
\rho^\sr_M\left(\psi^\sr\right)=\int_{K_M} \psi^\sr(\eta) \exp\left(\im S_M(\eta)\right)\,\xd\mu(\eta) .
\label{eq:pathint}
\end{equation}
Here $K_M$ is the ``space of field configurations'' in $M$ and $\mu$ is supposed to be a suitable measure on it. We shall assume at least that $K_M$ is a real vector space and that $\mu$ is translation-invariant. Of course, usually no such measure exists and even the precise definition of the space $K_M$ may be unclear. As a first step to improve the situation we assume that there is a correspondence between field configuration data on the boundary and solutions in the interior, i.e., $K_M$ splits additively into $K_M=A_M\oplus K_M^0$, where $A_M$ is the space of solutions of the equations of motion in $M$ while $K_M^0$ is the space of field configurations in $M$ that vanish on the boundary. Then, (\ref{eq:pathint}) may be rewritten as
\begin{equation}
\rho^\sr_M\left(\psi^\sr\right)=\int_{A_M} \psi^\sr(\eta) \left(\int_{K_M^0} \exp\left(\im S_M(\eta+\Delta)\right)\,\xd\mu(\Delta)\right)\xd\mu(\eta) ,
\label{eq:pathint2}
\end{equation}
where $\Delta\in K_M^0$ and the measure has now been split into one on $A_M$ and one on $K_M^0$.

To further improve the situation we use the fact that we are considering the special case of affine field theory. Thus, $A_M$ is an affine space, the action $S_M$ is a polynomial of degree two on $K_M$ and by the variational principle we obtain
\begin{equation}
S_M(\eta+\Delta)=S_M(\eta)+\Sq_M(\Delta)\quad\text{for}\quad\eta\in A_M,\; \Delta\in K_M^0 .
\end{equation}
Here $\Sq_M$ is the quadratic part of the action. In itself it is the action for the linear field theory with space of solutions in $M$ given by $L_M$. This allows to factorize the inner integrand in (\ref{eq:pathint2}), leading to the expression
\begin{equation}
\rho^\sr_M\left(\psi^\sr\right)= N_M \int_{A_M} \psi^\sr(\eta) \exp\left(\im S_M(\eta)\right)\,\xd\mu(\eta) ,
\label{eq:pathint3}
\end{equation}
where the normalization factor $N_M$ is given by
\begin{equation}
N_M=\int_{K_M^0} \exp\left(\im \Sq_M(\Delta)\right)\,\xd\mu(\Delta) .
\label{eq:nfact}
\end{equation}

Instead of evaluating the amplitude map (\ref{eq:pathint3}) for any possible state it is sufficient to consider coherent states since they generate a dense subspace of $\cH^\sr_{\partial M}$. Thus, consider the affine coherent state $\coha_\zeta$ associated to $\zeta\in A_{\partial M}$ with Schrödinger wave function $\cohas_\zeta$ given by expression (\ref{eq:acohs}). Its amplitude is
\begin{equation}
\rho^\sr_M\left(\cohas_\zeta\right)= N_M \int_{A_M} \cohas_\zeta(\eta) \exp\left(\im S_M(\eta)\right)\,\xd\mu(\eta) .
\label{eq:feynamplcoh}
\end{equation}
If $A_M$ is finite-dimensional this integral is perfectly well-defined and we shall proceed as if this was the case. Of course, usually in field theory $A_M$ is infinite-dimensional.

Inserting the explicit wave function (\ref{eq:acohs}) of the coherent state $\coha_\zeta$ yields,
\begin{multline}
\rho^\sr_M\left(\cohas_\zeta\right)= N_M \int_{A_{\tilde M}} \\
\exp\left(\im\, S_M(\eta)+\im\,\theta_{\partial M}(\zeta,\eta-\zeta)-\frac{1}{2}\Omega_{\partial M}(q_{\partial M}(\eta-\zeta),q_{\partial M}(\eta-\zeta))\right)\xd\mu(\eta) .
\label{eq:feynamplcohaev}
\end{multline}
Here we have also changed the integration from an integration over $A_M$ to an integration over its image $A_{\tilde{M}}$ under $a_M$ since the integrand only depends on the latter.

In order to evaluate this expression further we recall that the space $L_{\partial M}$ decomposes into a generalized direct sum $L_{\partial M}=A_{\tilde{M}}\oplus J_{\partial M} L_{\tilde{M}}$ \cite{Oe:affine}. Applying this to $\zeta$ we obtain $\zeta=\zeta^{\mathrm{R}}+J_{\partial M} \zeta^{\mathrm{I}}$ with $\zeta^{\mathrm{R}}\in A_{\tilde{M}}$ and $\zeta^{\mathrm{I}}\in L_{\tilde{M}}$. Inserting this decomposition into (\ref{eq:feynamplcohaev}), and using translation invariance of the measure $\mu$ to change the integral to one over the new variable $\xi\defeq \eta-\zeta^{\mathrm{R}}\in L_{\tilde M}$ we obtain
\begin{multline}
\rho^\sr_M\left(\cohas_\zeta\right) = N_M \int_{L_{\tilde M}}
\exp\bigg(\im\, S_M\left(\zeta^{\mathrm{R}}+\xi\right)+\im\,\theta_{\partial M}\left(\zeta^{\mathrm{R}}+J_{\partial M}\zeta^{\mathrm{I}},\xi-J_{\partial M} \zeta^{\mathrm{I}}\right) \\
\left. -\frac{1}{2}\Omega_{\partial M}\left(q_{\partial M}\left(\xi-J_{\partial M} \zeta^{\mathrm{I}}\right),q_{\partial M}\left(\xi-J_{\partial M} \zeta^{\mathrm{I}}\right)\right)\right)\xd\mu(\xi) .
\end{multline}
Using property (\ref{eq:actsympot}) of the action in axiom (C5) leads to,
\begin{multline}
\rho^\sr_M\left(\cohas_\zeta\right) = N_M \exp\left(\im\,S_M\left(\zeta^{\mathrm{R}}\right)-\im\,\theta_{\partial M}\left(\zeta^{\mathrm{R}},J_{\partial M} \zeta^{\mathrm{I}}\right)\right)\\
\int_{L_{\tilde M}}
\exp\left(-\frac{\im}{2} \left[\xi,\xi\right]_{\partial M}+\im\left[J_{\partial M}\zeta^{\mathrm{I}},\xi-J_{\partial M} \zeta^{\mathrm{I}}\right]_{\partial M} \right. \\
\left. -\frac{1}{2}\Omega_{\partial M}\left(q_{\partial M}\left(\xi-J_{\partial M} \zeta^{\mathrm{I}}\right),q_{\partial M}\left(\xi-J_{\partial M} \zeta^{\mathrm{I}}\right)\right)\right)\xd\mu(\xi) .
\end{multline}
We use again translation invariance of the measure $\mu$ to shift the integration variable by $\xi\to\xi+z \zeta^{\mathrm{I}}$, where $z\in\R$ is arbitrary. This yields,
\begin{multline}
\rho^\sr_M\left(\cohas_\zeta\right) = N_M \exp\left(\im\,S_M\left(\zeta^{\mathrm{R}}\right)-\im\,\theta_{\partial M}\left(\zeta^{\mathrm{R}},J_{\partial M} \zeta^{\mathrm{I}}\right)\right)\\
\int_{L_{\tilde M}}
\exp\left(-\frac{\im}{2} \left[\xi,\xi\right]_{\partial M}+\im\left[J_{\partial M}\zeta^{\mathrm{I}},\xi-J_{\partial M} \zeta^{\mathrm{I}}\right]_{\partial M} \right. \\
-\frac{1}{2}\Omega_{\partial M}\left(q_{\partial M}\left(\xi-J_{\partial M} \zeta^{\mathrm{I}}\right),q_{\partial M}\left(\xi-J_{\partial M} \zeta^{\mathrm{I}}\right)\right)\\
-\im z \left[\zeta^{\mathrm{I}},\xi\right]_{\partial M}
+\im z \left[J_{\partial M}\zeta^{\mathrm{I}},\zeta^{\mathrm{I}}\right]_{\partial M}
-z\,\Omega_{\partial M}\left(q_{\partial M}\left(\zeta^{\mathrm{I}}\right),q_{\partial M}\left(\xi-J_{\partial M} \zeta^{\mathrm{I}}\right)\right)\\
\left. -\frac{\im z^2}{2}\left[\zeta^{\mathrm{I}},\zeta^{\mathrm{I}}\right]_{\partial M}
- \frac{z^2}{2} \Omega_{\partial M}\left(q_{\partial M}\left(\zeta^{\mathrm{I}}\right),q_{\partial M}\left(\zeta^{\mathrm{I}}\right)\right)
\right)\xd\mu(\xi) .
\end{multline}
We make the simple observation that the integrand as a function of $z\in\C$ is holomorphic. This implies here that the integral is also holomorphic as a function of $z\in\C$. On the other hand, by construction the integral is constant for $z$ on the real line. Then, by the Identity Theorem, the integral has to be constant for all $z\in\C$. We may thus fix any convenient value $z\in\C$ to evaluate the integrand. It turns out that the convenient choice here is to set $z=-\im$, leading to considerable simplifications with the result,
\begin{multline}
\rho^\sr_M\left(\cohas_\zeta\right) = N_M N'_M\\ \exp\left(\im\,S_M\left(\zeta^{\mathrm{R}}\right)-\im\,\theta_{\partial M}\left(\zeta^{\mathrm{R}},J_{\partial M} \zeta^{\mathrm{I}}\right)-\frac{\im}{2}\left[J_{\partial M}\zeta^{\mathrm{I}},J_{\partial M} \zeta^{\mathrm{I}}\right]_{\partial M}-\frac{1}{2}g_{\partial M}\left(\zeta^{\mathrm{I}},\zeta^{\mathrm{I}}\right)\right) .
\label{eq:amplscoh}
\end{multline}
Here, the normalization factor $N'_M$ corresponds to the remaining integral,
\begin{equation}
N'_M=\int_{L_{\tilde M}}\exp\left(-\frac{\im}{2}[\xi,\xi]_\Sigma-\frac{1}{2}\Omega_\Sigma(q_\Sigma(\xi),q_\Sigma(\xi))\right)\xd\mu(\xi) .
\label{eq:npfact}
\end{equation}

Since the overall normalization in the Feynman path integral as considered here is a priori undetermined we may fix it in a convenient way. Indeed, the unique (region independent) normalization compatible with the core axioms turns out to be $N_M N'_M=1$. What is more, this makes the amplitude (\ref{eq:amplscoh}) coincide precisely with the amplitude of the holomorphic quantization scheme as given in Proposition~4.3 of \cite{Oe:affine}. More precisely, denoting by $\cH^\src_{\partial M}$ the subspace of $\cH^\sr_{\partial M}$ spanned by coherent states, we have $\rho_M^\sr=\rho_M^\hr\circ\bsta_{\partial M}$ on $\cH^\src_{\partial M}$. This implies the complete equivalence of the Schrödinger-Feynman quantization scheme considered here to the holomorphic quantization scheme proposed in \cite{Oe:affine}. In particular, the remaining core axioms (T4), (T3x), (T5a), (T5b) are also satisfied.

In the following we shall thus omit superscripts that distinguish the different quantization schemes in question and simply write $\cH_{\Sigma}$ for the Hilbert space associated to the hypersurface $\Sigma$ and $\rho_M$ for the amplitude map associated to the region $M$. We write the explicit expression of the amplitude on a coherent state $\coha_\zeta$, (\ref{eq:amplscoh}) as
\begin{multline}
\rho_M\left(\coha_\zeta\right) = \\ \exp\left(\im\,S_M\left(\zeta^{\mathrm{R}}\right)-\im\,\theta_{\partial M}\left(\zeta^{\mathrm{R}},J_{\partial M} \zeta^{\mathrm{I}}\right)-\frac{\im}{2}\left[J_{\partial M}\zeta^{\mathrm{I}},J_{\partial M} \zeta^{\mathrm{I}}\right]_{\partial M}-\frac{1}{2}g_{\partial M}\left(\zeta^{\mathrm{I}},\zeta^{\mathrm{I}}\right)\right) .
\label{eq:amplcoha}
\end{multline}

\subsection{Linear field theory}
\label{sec:linft}

Linear field theory arises as a special case of affine field theory, when we are given a choice of base point ``$0$'' in each space $A_M$ and $A_{\Sigma}$, in a compatible way.\footnote{Even the existence of such a choice is a non-trivial restriction. Recall the related discussion in \cite{Oe:affine}.} This allows to canonically identify $A_M$ with $L_M$ for every region $M$ and $A_{\Sigma}$ with $L_{\Sigma}$ for every hypersurface $\Sigma$. The axioms of the classical theory can then be considerably simplified by erasing all separate reference to the spaces $A_{\Sigma}$ and $A_M$. Also, we can then consistently set
\begin{equation}
\theta_\Sigma(\cdot,\cdot)\defeq [\cdot,\cdot]_\Sigma\quad\text{and}\quad S_M(\xi)\defeq -\frac{1}{2}[\xi,\xi]_{\Sigma} .
\label{eq:linsimpl}
\end{equation}
Indeed, it turns out that the explicit mention in the axioms of the symplectic potential $\theta_\Sigma$, its linearized version $[\cdot,\cdot]_{\Sigma}$ and the action $S_M$ is no longer required. We are left with the axioms for a linear field theory as given in \cite{Oe:holomorphic}.

Unsurprisingly, restricting to the special case of linear field theory preserves the equivalence of the Schrödinger-Feynman quantization scheme with the holomorphic quantization scheme. In the linear case the latter was first proposed in \cite{Oe:holomorphic}. For completeness, we recall that it is more convenient in the linear case to use the usual ``Fock space'' coherent states, or rather their normalized versions. The Schrödinger wave function for the normalized coherent state $\cohn_\tau\in\cH_\Sigma$ associated to $\tau\in L_{\Sigma}$ is then given by \cite{Oe:schroedhol},
\begin{equation}
\cohns_\tau(\phi)=\exp\left(\im [\tau,\phi]_{\Sigma}-\frac{\im}{2}[\tau,\tau]_{\Sigma}-\frac{1}{2}\Omega_{\Sigma}\left(\phi-q(\tau),\phi-q(\tau)\right)\right)
\label{eq:cohns}
\end{equation}
This is related to (\ref{eq:acohs}) through a choice of base point together with a $\tau$-dependent phase factor, see \cite{Oe:affine,Oe:schroedhol}. Indeed, given a base point $0\in A_{\Sigma}$ to identify $A_{\Sigma}$ with $L_{\Sigma}$ and the respective Hilbert spaces we obtain
\begin{equation}
\cohn_\tau=\coha_{\tau+0}\exp\left(\frac{\im}{2}[\tau,\tau]_{\Sigma}\right)
\label{eq:relcoh}
\end{equation}
One may then verify, either by repeating the calculation of Section~\ref{sec:feynquant} or by using this relation together with (\ref{eq:linsimpl}) on the result (\ref{eq:amplcoha}) that one obtains for the amplitude,
\begin{equation}
\rho_M\left(\cohn_\tau\right)=\exp\left(-\frac{\im}{2} g_{\partial M}\left(\tau^\mathrm{R},\tau^\mathrm{I}\right)-\frac{1}{2} g_{\partial M}\left(\tau^\mathrm{I},\tau^\mathrm{I}\right)\right),
\label{eq:amplcoh}
\end{equation}
in accordance with Proposition~4.2 in \cite{Oe:holomorphic}. Here, $\tau=\tau^\mathrm{R}+J_{\partial M}\tau^\mathrm{I}$ with $\tau^\mathrm{R},\tau^\mathrm{I}\in L_{\tilde{M}}$.


\section{Observables}
\label{sec:obs}

\subsection{Motivation: Observables in quantum field theory}
\label{sec:obsqft}

We recall in this section some basic facts about observables in quantum field theory that motivate much of the following treatment. The presentation here may be largely seen as a summary of the more extensive discussion in \cite{Oe:obsgbf}.

In non-relativistic quantum mechanics quantum observables are simply certain operators on the Hilbert space of the system. They are usually constructed through a quantization scheme from classical observables which are functions on the instantaneous phase space of the system. The operator can be applied at any time, usually indicating a measurement being performed at that time. Crucial information is contained in the commutation relations of these operators, indicating in particular different outcomes of joint measurements when the temporal order of constituent measurements is altered.

In quantum field theory the situation is quite different.\footnote{It should be emphasized that we refer here to the text-book approach to quantum field theory as it is used for example in high energy physics.} The standard elementary observables, such as field values $\phi(t,x)$ are \emph{labeled} not only with a position $x$ in space, but also with a time $t$. They are thus functions not on phase space, but on a larger space of \emph{field configurations} in spacetime. While the quantum observable corresponding to $\phi(x,t)$ is still thought of as an operator on the Hilbert space of the system, the physically relevant operation for combining such quantum observables does not require the full operator algebra structure. Rather, the physically relevant composition of quantum observables is given by the commutative \emph{time-ordered product}, which orders the operators according to their time labels. This strongly suggests that the operator point of view is not the most natural one here.

Indeed, the relation between the time-ordered product and the Feynman path integral suggests a different point of view. Suppose for simplicity that we are working with a real scalar field theory in Minkowski space. Consider a region $M=[t_a,t_b]\times \R^3$, where $t_a<t_b$ and denote by $K_{[t_a,t_b]}$ the space of field configurations in $M$. Consider a classical observable $K_{[t_a,t_b]}\to \R$ that encodes an $n$-point function,
\begin{equation}
 \phi\mapsto \phi(t_1,x_1)\cdots\phi(t_n,x_n),
\label{eq:npoint}
\end{equation}
where $t_1,\dots,t_n\in [t_a,t_b]$. Given an initial state $\psi\in\cH_{t_a}$ at time $t_a$ and a final state $\eta\in\cH_{t_b}$ at time $t_b$, the corresponding matrix element of the time-ordered product of (\ref{eq:npoint}) can be expressed by the Feynman path integral,\footnote{We use on the left-hand side a notation employing the Heisenberg picture, as is usual in standard text books on quantum field theory.}
\begin{multline}
\left\langle \eta, \tord\, \tilde{\phi}(t_1,x_1)\cdots\tilde{\phi}(t_n,x_n)\psi\right\rangle \\
=\int_{K_{[t_1,t_b]}} \psi(\phi|_{t_a})\overline{\eta(\phi|_{t_b})}\, \phi(t_1,x_1)\cdots\phi(t_n,x_n)\, \exp\left(\im S_{[t_a,t_b]}(\phi)\right)\,\xd\mu(\phi) .
\label{eq:tordnfeyn}
\end{multline}
Here, $\psi,\eta$ inside the integral are the Schrödinger wave functions of the respective states, $\tilde{\phi}(t_i,x_i)$ are the usual quantizations of the classical observables $\phi\mapsto \phi(t_i,x_i)$ and $\tord$ signifies time-ordering. When initial and final states are taken to be the vacuum, (\ref{eq:tordnfeyn}) recovers the usual quantum $n$-point function that is at the heart of the predictive power of quantum field theory.

The quantization occurring here may thus be seen as the conversion of a classical observable $F:K_{[t_a,t_b]}\to\R$ to a ``modified evolution operator'' $\hat{F}:\cH_{t_a}\to\cH_{t_b}$ with matrix elements,
\begin{equation}
\left\langle \eta, \hat{F} \psi\right\rangle \\
=\int_{K_{[t_a,t_b]}} \psi(\phi|_{t_a})\overline{\eta(\phi|_{t_b})}\, F(\phi)\, \exp\left({\im S_{[t_a,t_b]}(\phi)}\right)\,\xd\mu(\phi) .
\label{eq:tordfeyn}
\end{equation}
More in line with a GBF perspective we may also view this as a ``modified transition amplitude'' $\rho_{[t_a,t_b]}^F:\cH_{t_a}\tens\cH_{t_b}^*\to\C$,
\begin{equation}
\rho_{[t_a,t_b]}^F(\psi\tens\iota(\eta))=\left\langle \eta, \hat{F} \psi\right\rangle .
\end{equation}

A remarkable property of quantum field theory is a correspondence between the composition of classical and of quantum observables in this quantization prescription. This was termed \emph{composition correspondence} in \cite{Oe:obsgbf} and comes from a generic property of the path integral. Concretely, consider times $t_a<t_b<t_c$. Let $F:K_{[t_a,t_b]}\to\R$ be a classical observable in the region $[t_a,t_b]\times \R^3$ and $G:K_{[t_b,t_c]}\to\R$ be a classical observable in the region $[t_b,t_c]\times \R^3$. We can extend both $F$ and $G$ trivially to classical observables $K_{[t_a,t_c]}\to\R$ in the region $[t_a,t_c]\times \R^3$ and multiply them there as functions. We call the resulting observable $G\cdot F:K_{[t_a,t_c]}\to\R$. The prescription (\ref{eq:tordfeyn}) then leads to the identity
\begin{equation}
 \widehat{G\cdot F}=\hat{G}\circ\hat{F} .
\label{eq:compcorop}
\end{equation}
That is, there is a direct correspondence between the classical composition of observables (via multiplication of functions) and the quantum composition of observables (via multiplication of operators).

Note that \emph{locality properties} of the observables are crucial for the correspondence (\ref{eq:compcorop}). The prescription given above allows only observables with disjoint supports in spacetime to be composed in this way. Conversely, it is easy to see that a correspondence as in (\ref{eq:compcorop}) \emph{requires} the observables to be functions on some spacetime configuration space rather than on phase space (or the space of solutions). For suppose that observables were defined as functions on the space of solutions. Then, in the above example, the composed classical observable $G\cdot F$ ``forgets'' the spacetime localization of $F$ and $G$, leading essentially to an equality of the type $\widehat{G\cdot F}=\hat{F}\circ\hat{G}$ in addition to equation (\ref{eq:compcorop}).\footnote{There are some inessential subtleties to the argument that we are glossing over here, involving additional relative time-translations.} This in turn would essentially imply $\hat{G}\circ\hat{F}=\hat{F}\circ\hat{G}$, meaning that all quantum observables, viewed as operators commute, in contradiction to what we know about quantum (field) theory.

It is important to distinguish the composition correspondence (\ref{eq:compcorop}) from the much more established Dirac quantization condition relating the commutator of observables quantized as operators to the quantization of the Poisson bracket of the observables. In the latter, classical observables are necessarily understood as functions on phase space. For observables $D,E$ this takes the form
\begin{equation}
 \tilde{E}\circ\tilde{D}-\tilde{D}\circ\tilde{E}=-\im\, \widetilde{(E,D)},
\label{eq:diraccond}
\end{equation}
where the bracket indicates here the Poisson bracket and quantization is denoted with a tilde. In quantum field theory elementary observables at equal times can also be viewed as functions on phase space and as such realize the condition (\ref{eq:diraccond}) in the form of the canonical commutation relations.

\subsection{Observables in the GBF}
\label{sec:obsgbf}

The properties of quantum field theory discussed in the previous section are suggestive of a concept of quantum observable, introduced in \cite{Oe:obsgbf}, that naturally integrates into the GBF. We elaborate on this in the following.

A quantum observable $O$ is associated to a spacetime region $M$ and takes the form of a linear map, called \emph{observable map}
\begin{equation}
 O:\cH^\ds_{\partial M}\to\C ,
\end{equation}
similar to the amplitude map for the region $M$. The most important operation performed with observables is \emph{composition}, generalizing the temporal composition discussed in Section~\ref{sec:obsqft}. This is exactly analogous to the composition of amplitudes arising from the gluing of regions in the core axioms (T5a) and (T5b) of Section~\ref{sec:coreaxioms}. Also, to make the concept of observable useful it is necessary to consider the space of observables $\obs_M$ for each spacetime region $M$. This together with a closedness condition under composition is expressed through the following axioms, slightly modified from \cite{Oe:obsgbf}.
\begin{itemize}
\item[(O1)] Associated to each spacetime region $M$ is a real vector space $\obs_M$ of linear maps $\cH_{\partial M}^\ds\to\C$, called \emph{observable maps}. In particular, $\rho_M\in\obs_M$.
\item[(O2a)] Let $M_1$ and $M_2$ be regions and $M=M_1\cup M_2$ be their disjoint union. Then, there is an injective bilinear map $\aglue:\obs_{M_1}\times\obs_{M_2}\toi\obs_{M}$ such that for all $O_1\in\obs_{M_1}$ and $O_2\in\obs_{M_2}$ and $\psi_1\in\cH_{\partial M_1}^\ds$ and $\psi_2\in\cH_{\partial M_2}^\ds$,
\begin{equation}
 O_1\aglue O_2(\psi_1\tens\psi_2)= O_{1}(\psi_1)O_{2}(\psi_2) .
\end{equation}
This operation is required to be associative in the obvious way.
\item[(O2b)] Let $M$ be a region with its boundary decomposing as a disjoint union $\partial M=\Sigma_1\cup\Sigma\cup \overline{\Sigma'}$ and $M_1$ given as in (T5b). Then, there is a linear map $\aglue_{\Sigma}:\obs_{M}\to\obs_{M_1}$ such that for all $O\in\obs_{M}$ and any orthonormal basis $\{\xi_i\}_{i\in I}$ of $\cH_\Sigma$ and for all $\psi\in\cH_{\partial M_1}^\ds$,
\begin{equation}
 \aglue_{\Sigma}(O)(\psi)\cdot c(M;\Sigma,\overline{\Sigma'})
 =\sum_{i\in I}O(\psi\tens\xi_i\tens\iota_\Sigma(\xi_i)) .
\end{equation}
This operation is required to commute with itself and with (O2a) in the obvious way.
\end{itemize}

We note that the maps $\aglue$ and $\aglue_{\Sigma}$ are defined in such a way that equation (\ref{eq:glueaxa}) in axiom (T5a) can be rewritten as $\rho_M=\rho_{M_1}\aglue\rho_{M_2}$. Similarly, equation (\ref{eq:glueaxb}) in axiom (T5b) can be rewritten as $\rho_{M_1}=\aglue_{\Sigma}(\rho_M)$.

\subsection{Classical observables}
\label{sec:cobsax}

We provide in this section a minimal axiomatization of classical observables. We merely suppose that they form a commutative algebra for each region, and are subject to the natural operations coming from the gluing of regions.

\begin{itemize}
\item[(CO1)] Associated to each spacetime region $M$ is a real unital algebra $\cobs_M$, called \emph{observable algebra}.
\item[(CO2a)] Let $M_1$ and $M_2$ be regions and $M=M_1\cup M_2$ be their disjoint union. Then, there are injective algebra homomorphisms $l_{(M_1;M)}:\cobs_{M_1}\toi\cobs_M$ and $l_{(M_2;M)}:\cobs_{M_2}\toi\cobs_M$. This operation is required to be associative in the obvious way.
\item[(CO2b)] Let $M$ be a region with its boundary decomposing as a disjoint union $\partial M=\Sigma_1\cup\Sigma\cup \overline{\Sigma'}$ and $M_1$ given as in (T5b). Then, there is an algebra homomorphism $l_{(M;\Sigma,\overline{\Sigma'})}:\cobs_{M}\to\cobs_{M_1}$. This operation is required to commute with itself and with (CO2a) in the obvious way.
\end{itemize}

\subsection{Quantization axioms and composition correspondence}
\label{sec:compcor}

So far we have taken from Section~\ref{sec:obsqft} only the motivation for the general structure of quantum observables in the GBF. We shall now proceed to ``import'' and generalize the composition correspondence of quantum field theory, expressed in equation (\ref{eq:compcorop}), into the GBF. Indeed, switching from evolution operators to amplitudes, it is quite clear how this should be done. We shall give a formulation in line with the structure of the core axioms (T5a) and (T5b).
\begin{itemize}
\item[(X1)] Associated to each spacetime region $M$ there is a linear map $\quant_M:\cobs_M\to\obs_M$, called \emph{quantization map}. Moreover, $\quant_M(\one)=\rho_M$.
\item[(X2a)] Let $M_1$ and $M_2$ be regions and $M=M_1\cup M_2$ be their disjoint union. Then, the following diagram commutes.
\begin{equation}
\xymatrix{
  \cobs_{M_1}\times\cobs_{M_2} \ar[rr]^{\quant_{M_1}\times\quant_{M_2}} \ar[d]_{l_{(M_1;M)}\times l_{(M_2;M)}} & & \obs_{M_1}\times\obs_{M_2} \ar[dd]^{\aglue}\\
  \cobs_{M}\times\cobs_{M} \ar[d]_{\cdot} \\
  \cobs_{M} \ar[rr]_{\quant_M} & & \obs_{M}}
\end{equation}
\item[(X2b)] Let $M$ be a region with its boundary decomposing as a disjoint union $\partial M=\Sigma_1\cup\Sigma\cup \overline{\Sigma'}$ and $M_1$ given as in (T5b). Then, the following diagram commutes.
\begin{equation}
\xymatrix{
  \cobs_{M} \ar[rr]^{\quant_{M}} \ar[d]_{l_{(M;\Sigma,\overline{\Sigma'})}} & & \obs_{M} \ar[d]^{\aglue_{\Sigma}}\\
  \cobs_{M_1} \ar[rr]_{\quant_{M_1}} & & \obs_{M_1}}
\end{equation}
\end{itemize}

Axiom (X1) merely establishes the existence of a quantization map. The composition correspondence is encoded in the combination of axiom (X2a) and (X2b). In the concrete example of Section~\ref{sec:obsqft} this can be seen as follows. In a first step define $M_1=[t_a,t_b]\times\R^3$, $M_2=[t_b,t_c]\times\R^3$ and choose $F\in\cobs_{M_1}$ and $G\in\cobs_{M_2}$. We apply axiom (X2a). This leads to the formally disjoint union $M=[t_a,t_b]\times\R^3\cup [t_b,t_c]\times\R^3$. In a second step we apply axiom (X2b). That is, we glue $M$ to itself along the hypersurface $\Sigma={t_b}\times\R^3$ and its copy $\Sigma'$. Call the resulting region $M_3=[t_a,t_c]\times\R^3$. This yields the identity
\begin{equation}
 \quant_{M_3}\left(l_{(M;\Sigma,\overline{\Sigma'})}\left(l_{(M_1;M)}(F)\cdot l_{(M_2;M)}(G)\right)\right)=\aglue_{\Sigma}\left(\quant_{M_1}(F)\aglue\quant_{M_2}(G)\right) ,
\end{equation}
which is equivalent to (\ref{eq:compcorop}).

Contrary to superficial appearance, the essence of composition correspondence is contained in axiom (X2b) rather than in axiom (X2a). Indeed, axiom (X2a) may be interpreted as merely stating that quantization in disjoint regions is independent. This is a generic property that should be expected of any quantization scheme. Indeed, it is easy to verify axiom (X2a) explicitly for all quantization schemes for the GBF that were introduced in \cite{Oe:obsgbf}. In particular, this includes quantization schemes based on observables defined on spaces of solutions. As we saw already in Section~\ref{sec:obsqft} these cannot satisfy the composition correspondence. We shall come back to this in Section~\ref{sec:compquant}.

\subsection{Feynman quantization of Weyl observables}
\label{sec:feynweyl}

There is a further element of quantum field theory as discussed in Section~\ref{sec:obsgbf} that we wish to realize in the GBF. This is the concrete prescription for the quantization of observables via the Feynman path integral exhibited in expressions (\ref{eq:tordnfeyn}) and (\ref{eq:tordfeyn}). It is straightforwardly generalized to the GBF.

Let $M$ be a region, $K_M$ be the configuration space in $M$ and $F:K_M\to\R$ a classical observable. The quantization of $F$ is then the linear map $\rho_M^F:\cH_{\partial M}\to\C$ given by
\begin{equation}
\rho_M^F(\psi)=\int_{K_M} \psi\left(\phi|_{\partial M}\right) F(\phi) \exp\left(\im S_M(\phi)\right)\,\xd\mu(\phi) .
\label{eq:fquantobs}
\end{equation}
The difficulty here is of course in making this integral and its ingredients well defined. Recall, however, that we have made sense of just such an integral in Section~\ref{sec:feynquant} in a very special case. To take advantage of this we restrict from here onwards to affine field theory.

The second restriction we shall perform is to a special type of observable. Namely, we consider observables of the form
\begin{equation}
 F(\phi)=\exp\left(\im\, D(\phi)\right),
\label{eq:weylobs}
\end{equation}
where $D:K_M\to\R$ is a linear observable. We shall refer to observables of the form (\ref{eq:weylobs}) as \emph{Weyl observables}.\footnote{We take the name from the Weyl relations, which are about exponentials of imaginary quantized linear observables, see also Section~\ref{sec:diracquant}.} The path integral (\ref{eq:fquantobs}) then takes the form
\begin{equation}
\rho_M^F(\psi)=\int_{K_M} \psi\left(\phi|_{\partial M}\right) \exp\left(\im \left(S_M(\phi)+D(\phi)\right)\right)\,\xd\mu(\phi) .
\label{eq:fquantwobs}
\end{equation}
The peculiarity of this path integral is that we may view the sum $S_M+D$ as a new action. From this point of view we have a path integral without an observable, but for a new theory. However, the action $S_M$ determining an affine field theory is polynomial of degree two, so $S_M+D$ is still polynomial of degree two and determines thus another affine field theory. Thus, the path integral (\ref{eq:fquantwobs}) is just a version of the path integral (\ref{eq:pathint}). However, for $\psi$ a coherent state we have given a precise meaning to the latter in Section~\ref{sec:feynquant}.

Indeed, for an affine coherent state $\coha_\zeta\in\cH_{\partial M}$ with $\zeta\in A_{\partial M}$ the observable map (\ref{eq:fquantwobs}) should thus be given by expression (\ref{eq:amplcoha}), with the corresponding substitution of the action,
\begin{multline}
\rho_M^F\left(\coha_\zeta\right) =  \exp\bigg(\im\,S_M\left(\zeta^{\mathrm{R}}\right)+\im\, D\left(\zeta^{\mathrm{R}}\right) \\
 \left. -\im\,\theta_{\partial M}\left(\zeta^{\mathrm{R}},J_{\partial M} \zeta^{\mathrm{I}}\right)-\frac{\im}{2}\left[J_{\partial M}\zeta^{\mathrm{I}},J_{\partial M} \zeta^{\mathrm{I}}\right]_{\partial M}-\frac{1}{2}g_{\partial M}\left(\zeta^{\mathrm{I}},\zeta^{\mathrm{I}}\right)\right) .
\label{eq:wobscoha}
\end{multline}
However, recall that we fixed a normalization factor $N_M N_M'=1$ to arrive at this result. In Section~\ref{sec:feynquant} this was justified by the fact that neither $N_M$ nor $N_M'$ depended on $\zeta$. Here, however, we also need to take into account a possible dependence of the factors $N_M$ and $N_M'$ on the observable $D$ or on the space $A_M$ which in turn depends on $D$. However, inspecting equations (\ref{eq:nfact}) and (\ref{eq:npfact}) we see that there is no such dependence. ($\Sq_M$ in equation (\ref{eq:nfact}) corresponds to $S_M$ here.) This justifies the normalization also in the present case.

Note that we have assumed implicitly here that the structure of the theory at the boundary is unaltered by the addition of the linear observable $D$ to the action $S_M$. This is justified if the observable vanishes in a neighborhood of the boundary. Since $K_M$ is a space of local field configurations this requirement makes sense. On the other hand, if observables are sufficiently regular it should not matter if they are altered in an arbitrarily small neighborhood, suggesting that this requirement can be dropped. In any case this issue turns out not to matter from the axiomatic perspective taken in the following.

While it is possible to carry through all of the following in affine field theory, we restrict for simplicity to linear field theory. Then, we have the substitutions (\ref{eq:linsimpl}) and switch to the normalized ``Fock'' coherent states (\ref{eq:cohns}) via (\ref{eq:relcoh}). For the normalized coherent state associated to $\tau\in L_{\partial M}$ we thus obtain,
\begin{multline}
\rho_M^F\left(\cohn_\tau\right) =\\
  \exp\left(\im\, S_M\left(\tau^{\mathrm{R}}\right) + \im\, D\left(\tau^{\mathrm{R}}\right)+\frac{\im}{2}[\tau^{\mathrm{R}},\tau^{\mathrm{R}}]_{\partial M}
 -\frac{\im}{2}g_{\partial M}\left(\tau^{\mathrm{R}},\tau^{\mathrm{I}}\right)-\frac{1}{2}g_{\partial M}\left(\tau^{\mathrm{I}},\tau^{\mathrm{I}}\right)\right) .
\label{eq:wobscoh1}
\end{multline}
Despite superficial appearance, the difference to formula (\ref{eq:amplcoh}) is not only in an additional phase factor that only depends on $\tau^{\mathrm{R}}$. Crucially, the decomposition $\tau=\tau^{\mathrm{R}}+J_{\partial M}\tau^{\mathrm{I}}$ depends itself on $D$. This is because $\tau^{\mathrm{R}},\tau^{\mathrm{I}}$ are determined by the requirement $\tau^{\mathrm{R}}\in A_{\tilde{M}}^D$ and $\tau^{\mathrm{I}}\in L_{\tilde{M}}$. And here $A_{\tilde{M}}^D$ is the space of boundary images of solutions in $M$ of the theory with action $S_M+D$.

The expression (\ref{eq:wobscoh1}) can be further simplified by noticing that in the linear setting the action can be defined in terms of the other structures. For the action evaluated on solutions of the linear theory this was already used in Section~\ref{sec:linft} with the corresponding expression given in (\ref{eq:linsimpl}). In the present setting with observables this generalizes to,
\begin{equation}
 S_M(\eta)=-\frac{1}{2}[\eta,\eta]_{\partial M}-\frac{1}{2} D(\eta),
\label{eq:defact}
\end{equation}
where $\eta\in A_M^D$, i.e., $\eta$ is a solution of the theory with action $S_M+D$. We may take this as a definition of the action here, although it is possible to derive this equation under mild assumptions from Lagrangian field theory (compare equation (116) in \cite{Oe:affine}). We may now simplify expression (\ref{eq:wobscoh1}) to obtain,
\begin{equation}
\rho_M^F\left(\cohn_\tau\right) =  \exp\left(\frac{\im}{2} D\left(\tau^{\mathrm{R}}\right)
 -\frac{\im}{2}g_{\partial M}\left(\tau^{\mathrm{R}},\tau^{\mathrm{I}}\right)-\frac{1}{2}g_{\partial M}\left(\tau^{\mathrm{I}},\tau^{\mathrm{I}}\right)\right) .
\label{eq:wobscoh}
\end{equation}

There is an additional datum that we will need to capture axiomatically. Recall that in the affine theory axiom (C5) makes a crucial statement about the action in the form of relation (\ref{eq:actsympot}). In the present context this amounts to a statement not only about the action, but also about the linear observables $D$. Taking into account (\ref{eq:linsimpl}) it is easily seen that (\ref{eq:actsympot}) is here equivalent to,
\begin{equation}
 S_M(\eta+\xi)+D(\xi)=S_M(\eta)+ S_M(\xi)- [\eta,\xi]_{\partial M}\qquad\forall \eta\in A^D_M, \forall \xi\in L_M ,
\label{eq:actshift}
\end{equation}
where $A^D_M$ is the space of solutions in $M$ of the theory with action $S_M+D$. Comparison with expression (\ref{eq:defact}) for the action, evaluated on $\eta$ on the one hand and on $\eta+\xi$ on the other hand yields the identity
\begin{equation}
 D(\xi)=2\omega_{\partial M}(\xi,\eta)\qquad\forall\eta\in A_M^D,\forall \xi\in L_M .
\label{eq:losympl}
\end{equation}
This might look surprising at first since the left hand side does not depend on $\eta$. However, we recall that any other element $\eta'\in A_M^D$ is related to $\eta$ by an element of $L_{M}$ which does not contribute in the symplectic form with another element of $L_M$ since $L_M$ is isotropic in $L_{\partial M}$. It is easy to see that given (\ref{eq:defact}) the relations (\ref{eq:actshift}) and (\ref{eq:losympl}) are actually equivalent.

Note also that property (\ref{eq:losympl}) implies that $D$ evaluated on any configuration that is a solution associated with some linear observable only depends on the boundary image in $L_{\partial M}$ of that configuration. In particular the first term in the exponential on the right-hand side of (\ref{eq:wobscoh}) is well defined. Another consequence of (\ref{eq:losympl}) is a continuity property of $D$: Noting the equality
\begin{equation}
 D(\xi)=2\omega_{\partial M}(\xi,\eta)=g_{\partial M}(\xi,-J_{\partial M}\eta)
\end{equation}
we see that $D$ is continuous on the space $L_M\subseteq K_M$ with topology induced from $L_{\tilde M}$, by the Riesz representation theorem for real Hilbert spaces.

\subsection{Encoding classical linear field theory with linear observables}
\label{sec:classlinobs}

To formalize the setting of Section~\ref{sec:feynweyl} on the classical side we have to add two ingredients to the axioms of linear classical field theory as given in \cite{Oe:holomorphic} (or as obtained by reducing the axioms of Section~\ref{sec:classaffine} according to Section~\ref{sec:linft}). The first is the space of field configurations $K_M$ for each region $M$. The second is a space $\cobsl_M$ of linear observables $K_M\to\R$ for each region $M$. Moreover, recalling the discussion of the observable map (\ref{eq:wobscoh}), we need to somehow assign its space $A_M$ of modified solutions in $M$ to every linear observable. To this end it makes sense to assume that $A_M$ is a subspace of $K_M$. Moreover, recall that $A_M$ is an affine space over the vector space $L_M$ of solutions of the linear theory in $M$. Thus $A_M$, corresponds to a point in the quotient space $K_M/L_M$.

Using generic properties of the action including the variational principle for classical solutions it is easy to see the following: Given a solution $\eta_1$ of the the theory determined by the action $S_M+D_1$ and a solution $\eta_2$ of the theory determined by the action $S_M+D_2$ the configuration $\eta_1+\eta_2$ is a solution of the theory determined by the action $S_M+D_1+D_2$. This implies that the map $\cobsl_M\to K_M/L_M$ considered above is linear.

\begin{itemize}
\item[(C1)] Associated to each hypersurface $\Sigma$ is a complex separable Hilbert space $L_\Sigma$. The inner product in $L_\Sigma$ is denoted by $\{\cdot ,\cdot\}_\Sigma$. We also define $g_\Sigma(\cdot,\cdot)\defeq \Re\{\cdot ,\cdot\}_\Sigma$ and $\omega_\Sigma(\cdot,\cdot)\defeq \frac{1}{2}\Im\{\cdot ,\cdot\}_\Sigma$ and denote by $J_\Sigma:L_\Sigma\to L_\Sigma$ the scalar multiplication with $\im$ in $L_\Sigma$. Moreover we suppose there is a continuous bilinear map $[\cdot,\cdot]_\Sigma:L_\Sigma\times L_\Sigma\to\R$ such that
\begin{equation}
 \omega_\Sigma(\phi,\phi')=\frac{1}{2} [\phi,\phi']_\Sigma-\frac{1}{2} [\phi',\phi]_\Sigma
\qquad \forall  \phi,\phi'\in L_\Sigma .
\end{equation}
\item[(C2)] Associated to each hypersurface $\Sigma$ there is a conjugate linear involution $L_\Sigma\to L_{\overline\Sigma}$ under which the inner product is complex conjugated. We will not write these maps explicitly, but rather think of $L_\Sigma$ as identified with $L_{\overline\Sigma}$. Then, $\{\phi',\phi\}_{\overline{\Sigma}}=\overline{\{\phi',\phi\}_\Sigma}$ and we also require $[\phi,\phi']_{\overline{\Sigma}}=-[\phi,\phi']_\Sigma$ for all $\phi,\phi'\in L_\Sigma$.
\item[(C3)] Suppose the hypersurface $\Sigma$ decomposes into a disjoint
  union of hypersurfaces $\Sigma=\Sigma_1\cup\cdots\cup\Sigma_n$. Then,
  there is an isometric isomorphism of complex Hilbert spaces
  $L_{\Sigma_1}\oplus\cdots\oplus L_{\Sigma_n}\to L_\Sigma$. Moreover, these maps satisfy obvious associativity conditions. We will not write these maps explicitly, but rather think of them as identifications. Also, $[\cdot,\cdot]_\Sigma=[\cdot,\cdot]_{\Sigma_1}+\dots+[\cdot,\cdot]_{\Sigma_n}$.
\item[(C4)] Associated to each region $M$ is a real vector space $K_M$ and a subspace $L_M$. Also, there is a real vector space $\cobsl_M$ of linear maps $K_M\to\R$ as well as an injective linear map $s_M:\cobsl_M\toi K_M/L_M$.
\item[(C5)] Associated to each region $M$ there is a linear map of real vector spaces $k_M:K_M\to L_{\partial M}$. We denote by $L_{\tilde{M}}$ the image of the subspace $L_M$ under $k_M$. $L_{\tilde{M}}$ is a closed Lagrangian subspace of the real Hilbert space $L_{\partial M}$ with respect to the symplectic form $\omega_{\partial M}$. We often omit the explicit mention of the maps $k_M$. We also require for all $D\in\cobsl_M$ that
\begin{equation}
D(\xi)=2\omega_{\partial M}(\xi,\eta)\quad\text{if}\quad \xi\in L_M\quad\text{and}\quad \eta\in s_M(D) .
\label{eq:lobprop}
\end{equation}
\item[(C6)] Let $M_1$ and $M_2$ be regions and $M\defeq M_1\cup M_2$ be their disjoint union. Then, there is an isomorphism of real vector spaces $K_{M_1}\times K_{M_2}\to K_M$ restricting to an isomorphism of subspaces $L_{M_1}\times L_{M_2}\to L_M$ and inducing an isomorphism of complex vector spaces $L_{\partial M_1}\times L_{\partial M_2}\to L_{\partial M}$ such that the following diagram commutes.
\begin{equation}
\xymatrix{
  K_{M_1}\times K_{M_2} \ar[rr] \ar[d]_{k_{M_1}\times k_{M_2}} & & K_{M} \ar[d]^{k_{M}}\\
  L_{\partial M_1}\times L_{\partial M_2}\ar[rr]  & & L_{\partial M}}
\end{equation}
Also, there are induced linear injections $\cobsl_{M_1}\toi\cobsl_{M}$ and $\cobsl_{M_2}\toi\cobsl_{M}$ in the obvious way. Moreover, the following diagram commutes for $i=1,2$.
\begin{equation}
\xymatrix{
  \cobsl_{M_i} \ar[rr]^{s_{M_i}} \ar[d] & & K_{M_i}/L_{M_i} \ar[d] \\
  \cobsl_{M}\ar[rr]_{s_M}  & & K_M/L_{M}}
\end{equation}
Furthermore, all these maps satisfy obvious associativity conditions.
\item[(C7)] Let $M$ be a region with its boundary decomposing as a disjoint union $\partial M=\Sigma_1\cup\Sigma\cup \overline{\Sigma'}$, where $\Sigma'$ is a copy of $\Sigma$. Let $M_1$ denote the gluing of $M$ to itself along $\Sigma,\overline{\Sigma'}$ and suppose that $M_1$ is a region. Note $\partial M_1=\Sigma_1$. Then, there is an injective linear map $k_{M;\Sigma,\overline{\Sigma'}}:K_{M_1}\toi K_{M}$ such that $k_{M;\Sigma,\overline{\Sigma'}}(L_{M_1})\subseteq L_M$ and
\begin{equation}
 K_{M_1}\toi K_{M}\rightrightarrows L_\Sigma\quad\text{as well as}\quad L_{M_1}\toi L_{M}\rightrightarrows L_\Sigma
\label{eq:c7d1}
\end{equation}
are both exact sequences. Here, the arrows on the right hand sides are compositions of the map $k_M$ with the projections of $L_{\partial M}$ to $L_\Sigma$ and $L_{\overline{\Sigma'}}$ respectively (the latter identified with $L_\Sigma$). 
Moreover, the following diagram commutes, where the bottom arrow is the orthogonal projection.
\begin{equation}
\xymatrix{
  K_{M_1} \ar[rr]^{k_{M;\Sigma,\overline{\Sigma'}}} \ar[d]_{k_{M_1}} & & K_{M} \ar[d]^{k_{M}}\\
  L_{\partial M_1}  & & L_{\partial M} \ar[ll]}
\label{eq:c7d2}
\end{equation}
Also, there is an induced linear map $\cobsl_{M}\to\cobsl_{M_1}$ in the obvious way and the following diagram commutes.
\begin{equation}
\xymatrix{
  \cobsl_{M} \ar[rr]^{s_{M}} \ar[d] & & K_{M}/L_{M} \\
   \cobsl_{M_1} \ar[rr]_{s_{M_1}}  & & K_{M_1}/L_{M_1} \ar[u]}
\label{eq:c7d3}
\end{equation}
\end{itemize}

Note that the spaces of linear observables we have considered here do not satisfy the axioms of Section~\ref{sec:cobsax}. This is not surprising since they are not closed under multiplication. However, the observables we actually want to quantize are not the linear observables, but the Weyl observables. We define for each region $M$,
\begin{equation}
\cobsw_M\defeq \left\{\phi\mapsto \exp(\im\, D(\phi)) : D\in\cobsl_M\right\} .
\end{equation}
These are not algebras either, but rather multiplicative groups. We can make them into algebras, however, by allowing linear combinations. Thus, we define
\begin{equation}
\cobs_M\defeq\left\{\sum_{i=1}^n \lambda_i F_i : \lambda_i\in\C, F_i\in \cobsw_M\right\} .
\end{equation}
The spaces $\cobs_M$ so defined do satisfy the axioms (CO1), (CO2a), (CO2b) of Section~\ref{sec:cobsax} with the detail that they are actually complex algebras of complex observables rather than real algebras of real observables.

\subsection{GBQFT of linear field theory with Weyl observables}
\label{sec:gbqftlow}

The quantization of a classical field theory described in terms of the axioms of Section~\ref{sec:classlinobs} is now straightforward. The Hilbert spaces associated to hypersurfaces and the amplitude maps associated to regions are as described in Section~\ref{sec:sfaffine}, specialized to linear field theory, see in particular Section~\ref{sec:linft}. As already shown in \cite{Oe:holomorphic} these satisfy the core axioms of the GBF (Section~\ref{sec:coreaxioms}).

The new ingredients are the observables with their quantization performed as described in Section~\ref{sec:feynweyl}. That is, for each region $M$ we define a linear map $\quant_M:\cobs_M\to\obs_M$ as follows. Since $\cobs_M$ is the space of linear combinations of Weyl observables, it is sufficient to define $\quant_M$ on those and extend it to $\cobs_M$ as a complex linear map. Now, for $F\in\cobsw_M$ there is $D\in\cobsl_M$ such that $F=\exp\left(\im\, D\right)$. We then define $\quant_M(F)\defeq \rho_M^F$, where $\rho_M^{F}:\cH_{\partial M}^\ds\to\C$ is given on coherent states by expression (\ref{eq:wobscoh}) as follows. First note that by axiom (C4) associated to $D$ is an element $s_M(D)$ in the quotient space $K_M/L_M$. As explained previously, this is equivalent to an affine subspace of $K_M$ which we shall denote by $A_M^D$. We denote its image in $L_{\partial M}$ under $k_M$ by $A_{\tilde{M}}^D$. Then, $L_{\partial M}$ can be decomposed as a generalized direct sum $L_{\partial M}=A_{\tilde{M}}^D\oplus J_{\partial M} L_{\tilde M}$ (compare Lemma~3.2 in \cite{Oe:affine}). That is, given $\tau\in L_{\partial M}$, there are unique elements $\tau^{\mathrm{R}}\in A_{\tilde M}^D$ and $\tau^{\mathrm{I}}\in L_{\tilde M}$ such that $\tau=\tau^{\mathrm{R}}+ J_{\partial M} \tau^{\mathrm{I}}$. Then, we define
\begin{equation}
\rho_M^F\left(\cohn_\tau\right) \defeq  \exp\left(\frac{\im}{2} D\left(\tau^{\mathrm{R}}\right)
 -\frac{\im}{2}g_{\partial M}\left(\tau^{\mathrm{R}},\tau^{\mathrm{I}}\right)-\frac{1}{2}g_{\partial M}\left(\tau^{\mathrm{I}},\tau^{\mathrm{I}}\right)\right) .
\label{eq:wobscohdef}
\end{equation}

This yields a definition of the quantization map $\quant_M$ for the region $M$ as well as of the set of observable maps $\obs_M$ as its image. Moreover, the constant function with value $1$, $F=\one$ is a Weyl observable obtained from the linear function with value $0$ for which formula (\ref{eq:wobscohdef}) just yields the ordinary amplitude (\ref{eq:amplcoh}). In particular, $\one\in\cobs_M$ and $\quant_M(\one)=\rho_M$. Thus axioms (O1) and (X1) are satisfied.

Consider now two regions $M_1,M_2$ and their disjoint union $M=M_1\cup M_2$. Let $D_1\in\cobsl_{M_1}$ and $D_2\in\cobsl_{M_2}$. According to axiom (C6) the linear observables obtained by extending the domains of $D_1$ and $D_2$ from $K_{M_1}$ and $K_{M_2}$ to $K_M=K_{M_1}\times K_{M_2}$ in the obvious way are contained in $\cobs_M$ and so is thus their sum, which we shall denote by $D$. Explicitly,
\begin{equation}
 D(\phi_1,\phi_2)=D_1(\phi_1)+D_2(\phi_2)\qquad\forall \phi_1\in K_{M_1},\,\forall\phi_2\in K_{M_2} .
\end{equation}
The corresponding Weyl observables $F_1=\exp\left(\im\, D_1\right)$, $F_2=\exp\left(\im\, D_2\right)$ and $F=\exp\left(\im\, D\right)$ are thus elements of the respective spaces $F_1\in\cobsw_{M_1}$, $F_2\in\cobsw_{M_2}$, $F\in\cobsw_M$ and we have
\begin{equation}
F(\phi_1,\phi_2)=F_1(\phi_1)\cdot F_2(\phi_2)\qquad\forall \phi_1\in K_{M_1},\,\forall\phi_2\in K_{M_2} .
\end{equation}
Now by axiom (C6) we have $A_{\tilde M}^{D}=A_{\tilde M_1}^{D_1}\times A_{\tilde M_2}^{D_2}$. In particular the decomposition $L_{\partial M}=A_{\tilde{M}}^D\oplus J_{\partial M} L_{\tilde M}$ splits into a corresponding decomposition for each of the regions $M_1$, $M_2$ separately. We also recall \cite{Oe:holomorphic} that given $\tau=\tau_1+\tau_2$ with $\tau\in L_{\partial M}$, $\tau_1\in L_{\partial M_1}$, $\tau_2\in L_{\partial M_2}$ the normalized coherent state $\cohn_\tau\in\cH_{\partial M}$ factorizes as
\begin{equation}
\cohn_\tau=\cohn_{\tau_1}\tens\cohn_{\tau_2}
\end{equation}
with $\cohn_{\tau_1}\in\cH_{\partial M_1}$ and $\cohn_{\tau_2}\in\cH_{\partial M_2}$. Thus, the observable map (\ref{eq:wobscohdef}) obtained by quantizing $F$ factorizes completely,
\begin{equation}
\rho_M^F(\cohn_{\tau_1}\tens\cohn_{\tau_2})=\rho_M^F(\cohn_{\tau})
 =\rho_{M_1}^{F_1}(\cohn_{\tau_1})\rho_{M_2}^{F_2}(\cohn_{\tau_2})
 =\left(\rho_{M_1}^{F_1}\aglue\rho_{M_2}^{F_2}\right)(\cohn_{\tau_1}\tens\cohn_{\tau_2}) .
\end{equation}
Since coherent states are dense in the boundary state spaces and Weyl observables generate the observable algebras as vector spaces, this is sufficient to prove that axioms (O2a) and (X2a) hold.

As explained in Section~\ref{sec:compcor} the most interesting and non-trivial axiom is (X2b) as this encodes composition correspondence. We proceed to prove it in the present context, together with (O2b). First we need a special identity.
\begin{lem}
\label{lem:relomampl}
Let $M$ be a region and $D\in\cobsl_M$. Define $F\in\cobsw_M$ as $F\defeq\exp\left(\im\, D\right)$.
Given $\eta\in A_{\tilde M}^D$ and $\xi\in L_{\partial M}$ the following identity holds,
\begin{equation}
 \rho_M^F\left(\cohn_{\eta+\xi}\right)=
 \rho_M\left(\cohn_{\xi}\right)\exp\left(\frac{\im}{2} D(\eta)+\im\,\omega_{\partial M}\left(\xi,\eta\right)\right) .
\end{equation}
\end{lem}
\begin{proof}
Let $\xi=\xi^{\mathrm{R}}+J_{\partial M}\xi^{\mathrm{I}}$ be the decomposition of $\xi$ such that $\xi^{\mathrm{R}}, \xi^{\mathrm{I}}\in L_{\tilde M}$. Then, $\eta+\xi=(\eta+\xi^{\mathrm{R}})+J_{\partial M}\xi^{\mathrm{I}}$ with $\eta+\xi^{\mathrm{R}}\in A_{\tilde{M}}^D$. Thus, with expression (\ref{eq:wobscohdef}) we obtain,
\begin{align}
 \rho_M^F\left(\cohn_{\eta+\xi}\right) & =
 \exp\left(\frac{\im}{2}D\left(\eta+\xi^{\mathrm{R}}\right)-\frac{\im}{2}g_{\partial M}\left(\eta+\xi^{\mathrm{R}},\xi^{\mathrm{I}}\right)-\frac{1}{2}g_{\partial M}\left(\xi^{\mathrm{I}},\xi^{\mathrm{I}}\right)\right) \\
& = \rho_M\left(\cohn_{\xi}\right)\exp\left(\frac{\im}{2} D\left(\eta+\xi^{\mathrm{R}}\right)-\frac{\im}{2}g_{\partial M}\left(\eta,\xi^{\mathrm{I}}\right)\right) \\
& = \rho_M\left(\cohn_{\xi}\right)\exp\left(\frac{\im}{2} D\left(\eta\right)+\im\,\omega_{\partial M}\left(\xi^{\mathrm{R}},\eta\right)-\frac{\im}{2}g_{\partial M}\left(\eta,\xi^{\mathrm{I}}\right)\right) \\
& = \rho_M\left(\cohn_{\xi}\right)\exp\left(\frac{\im}{2} D\left(\eta\right)+\im\,\omega_{\partial M}\left(\xi,\eta\right)\right) .
\end{align}
Here we have used the identity (\ref{eq:lobprop}) of axiom (C5).
\end{proof}

\begin{prop}
Let $M$ be a region with its boundary decomposing as a disjoint union $\partial M=\Sigma_1\cup\Sigma\cup \overline{\Sigma'}$ and $M_1$ given as in (T5b). Let $D\in\cobsl_M$ and define $F\in\cobsw_M$ by $F=\exp\left(\im\, D\right)$. Also we denote by $D_1$ the induced element in $\cobsl_{M_1}$ and define $F_1\in\cobsw_{M_1}$ by $F_1=\exp\left(\im\, D_1\right)$. Given, moreover, an orthonormal basis $\{\xi_i\}_{i\in I}$ of $\cH_{\Sigma}$ we have for all $\psi\in \cH^{\ds}_{\Sigma_1}$,\footnote{We refer to \cite{Oe:holomorphic} for details concerning the integral and related notation.}
\begin{equation}
\rho_{M_1}^{F_1}\left(\psi\right)\cdot c(M;\Sigma,\overline{\Sigma'})
= \sum_{i\in I} \rho_M^F\left(\psi\tens\xi_i\tens\iota_{\Sigma}(\xi_i)\right) .
\label{eq:obsglueid}
\end{equation}
\end{prop}
\begin{proof}
Since the space spanned by coherent states is dense in $\cH^{\ds}_{\Sigma_1}$ it will be enough to take $\psi$ to be a coherent state. To this end choose $\eta_1\in A_{M_1}^{D_1}\defeq s_{M_1}(D_1)$. According to diagram (\ref{eq:c7d3}) in axiom (C7) there is $\eta\in A_{M_1}^{D}\defeq s_M(D)$ such that $\eta=k_{M;\Sigma,\overline{\Sigma'}}(\eta_1)$. That is, we choose a solution $\eta_1$ of the theory determined by $S_{M_1}+D_1$ in $M_1$. This induces a solution $\eta$ of the theory with action $S_M+D$ in $M$. This should really be thought of as the same solution, just living in the larger configuration space $K_M$ rather than in $K_{M_1}$. Indeed, diagram (\ref{eq:c7d3}) ensures that the boundary image on $\Sigma_1$ is the same for $\eta$ and $\eta_1$. We denote this boundary image by $\tilde{\eta}_1$. Similarly, we denote the boundary images of $\eta$ on $\Sigma$ and $\Sigma'$ by $\tilde{\eta}_{\Sigma}$ and $\tilde{\eta}_{\Sigma'}$. The left hand diagram of (\ref{eq:c7d1}) then ensures the intuitively obvious equality $\tilde{\eta}_{\Sigma}=\tilde{\eta}_{\Sigma'}$. We then obtain for all $\delta\in L_{\Sigma_1}$ the following equality, equivalent to (\ref{eq:obsglueid}).
\begin{align}
& \sum_{i\in I} \rho_M^F\left(\cohn_{\tilde{\eta}_1+\delta}\tens\xi_i\tens\iota_\Sigma(\xi_i)\right) \label{eq:gid1}\\
& = \int_{\hat{L}_{\Sigma}} \rho_M^F\left(\cohn_{\tilde{\eta}_1+\delta}\tens\cohn_\tau\tens\cohn_\tau\right)\exp\left(\frac{1}{2}g_{\Sigma}(\tau,\tau)\right)\xd\nu_\Sigma(\tau) \label{eq:gid2}\\
& = \int_{\hat{L}_{\Sigma}} \rho_M^F\left(\cohn_{\tilde{\eta}_1+\delta}\tens\cohn_{\tilde{\eta}_\Sigma+\tau}\tens\cohn_{\tilde{\eta}_\Sigma+\tau}\right)\exp\left(\frac{1}{2}g_{\Sigma}(\tau,\tau)\right)\,\xd\nu_\Sigma(\tau) \label{eq:gid3}\\
& = \exp\left(\frac{\im}{2} D\left(\eta\right)+\im\,\omega_{\Sigma_1}\left(\delta,\tilde{\eta}_1\right)\right)\nonumber\\
&\quad
\int_{\hat{L}_{\Sigma}} \rho_M\left(\cohn_{\delta}\tens\cohn_{\tau}\tens\cohn_{\tau}\right)\exp\left(\frac{1}{2}g_{\Sigma}(\tau,\tau)\right)\xd\nu_\Sigma(\tau) \label{eq:gid4}\\
& = \exp\left(\frac{\im}{2} D\left(\eta\right)+\im\,\omega_{\Sigma_1}\left(\delta,\tilde{\eta}_1\right)\right) \sum_{i\in I} \rho_M\left(\cohn_{\delta}\tens\xi_i\tens\iota_\Sigma(\xi_i)\right) \label{eq:gid5}\\
& = \exp\left(\frac{\im}{2} D_1\left(\eta_1\right)+\im\,\omega_{\Sigma_1}\left(\delta,\tilde{\eta}_1\right)\right) \rho_{M_1}\left(\cohn_\delta\right) c(M;\Sigma,\overline{\Sigma'}) \label{eq:gid6}\\
& = \rho_{M_1}^{F_1}\left(\cohn_{\tilde{\eta}_1+\delta}\right) c(M;\Sigma,\overline{\Sigma'}) . \label{eq:gid7}
\end{align}
We recall from \cite{Oe:holomorphic} that the sum over an orthogonal basis in (\ref{eq:gid1}) can be replaced by  an integral over coherent states. For normalized coherent states this takes the form (\ref{eq:gid2}). We refer the reader to \cite{Oe:holomorphic} for more details about the integral involved as well as the notation used. The step from (\ref{eq:gid2}) to (\ref{eq:gid3}) is an application of Proposition~3.11 of \cite{Oe:holomorphic}, that consists in a ``shifting of the integrand''. Applying Lemma~\ref{lem:relomampl} then yields expression (\ref{eq:gid4}). Note here that the contribution from $\Sigma$ and $\Sigma'$ to $\omega_{\partial M}$ cancel each other due to the opposite orientations, leaving only a $\omega_{\Sigma_1}$-term. We proceed to replace the integral with a sum over an orthonormal basis to obtain (\ref{eq:gid5}). Applying now axiom (T5b) yields (\ref{eq:gid6}). Here we have also used $D(\eta)=D_1(\eta_1)$. Finally, applying again Lemma~\ref{lem:relomampl} yields (\ref{eq:gid7}).
\end{proof}

We note that in terms of the notation introduced in axiom (O2b), expression (\ref{eq:obsglueid}) can be written as,
\begin{equation}
\rho_{M_1}^{F_1}
= \aglue_{\Sigma}\left(\rho_M^F\right) .
\label{eq:obsglueax}
\end{equation}
By linearity, this equation holds for all observables $F\in\cobs_{M}$ and not only for the Weyl observables. Thus, equation (\ref{eq:obsglueax}) is precisely equivalent to the commutative diagram of (X2b), demonstrating the validity of the axiom. The validity of the much weaker statement of axiom (O2b) is then implied.

This concludes the rigorous demonstration that the GBQFT with observables obtained by applying the proposed quantization scheme to the axiom system of Section~\ref{sec:classlinobs} satisfies in addition to the core axioms (Section~\ref{sec:coreaxioms}) not only the observable axioms (Section~\ref{sec:obsgbf}), but also the quantization axioms (Section~\ref{sec:compcor}), including the important principle of composition correspondence.

\subsection{Factorization}
\label{sec:factorization}

A powerful tool in deriving properties of quantum field theory (for example Feynman rules) is the ``generating function'' or ``kernel'' of the S-matrix, see e.g.\ \cite{FaSl:gaugeqft,ItZu:qft}. This is essentially the S-matrix of free quantum field theory modified by a source field and evaluated between initial and final coherent states. Using the notation for a real Klein-Gordon theory for simplicity this takes the form
\begin{multline}
\langle \cohn_{\tau_{\mathrm{out}}}, S_\mu\, \cohn_{\tau_{\mathrm{in}}}\rangle
 =\langle \cohn_{\tau_{\mathrm{out}}}, \cohn_{\tau_{\mathrm{in}}}\rangle \\
\exp\left(\im \int \mu(x) \hat{\tau}(x)\,\xd x\right)
 \exp\left(\frac{\im}{2}\int \mu(x) G_F(x,x') \mu(x')\,\xd x\xd x'\right) .
\label{eq:genfsmatrix}
\end{multline}
Here $\mu$ is the source field, $G_F$ is the Feynman propagator and $\hat{\tau}$ is a complex solution of the Klein-Gordon equation determined by initial and final boundary conditions given by $\tau_{\mathrm{in}}$ and $\tau_{\mathrm{out}}$ respectively.

From a GBF point of view $S_\mu$ determines an amplitude $\rho_M^\mu$ of the theory with source. For comparison we denote the amplitude for the theory without source by $\rho_M$. The coherent states can be combined into a single boundary coherent state $\cohn_{\tau}=\cohn_{\tau_{\mathrm{in}}}\tens \cohn_{\tau_{\mathrm{out}}}$. Equality (\ref{eq:genfsmatrix}) thus takes the form
\begin{multline}
 \rho_M^\mu(\cohn_{\tau})
 =\rho_M(\cohn_{\tau}) \\
\exp\left(\im \int \mu(x) \hat{\tau}(x)\,\xd x\right)
 \exp\left(\frac{\im}{2}\int \mu(x) G_F(x,x') \mu(x')\,\xd x\xd x'\right) .
\label{eq:genfaampl}
\end{multline}
Here we may think of $M$ as determined by the asymptotic limit $t_{\mathrm{in}}\to-\infty$, $t_{\mathrm{out}}\to \infty$ of the region $[t_{\mathrm{in}},t_{\mathrm{out}}]\times\R^3$ in Minkowski space.

Intriguingly, it was found in \cite{CoOe:spsmatrix,CoOe:smatrixgbf} that the very same expression (\ref{eq:genfaampl}) describes the (asymptotic limit for $R\to\infty$ of the) amplitude for a region of the form $M=\R\times B^3_R$. Here, $B^3_R$ denotes the ball of radius $R$ in space, centered at the origin. This was at the heart there of the proof that this asymptotic amplitude is equivalent to the usual S-matrix in Minkowski space.

Indeed, as was already shown in \cite{Oe:affine} (from a slightly different perspective than the one we shall take here) the identity (\ref{eq:genfaampl}) is merely the incarnation of a much more general identity in the GBF. To see this we think of the source field $\mu$ as determining a linear function from field configurations to the real numbers,
\begin{equation}
D(\phi)\defeq \int \mu(x) \phi(x)\,\xd x .
\label{eq:srcobs}
\end{equation}
While in \cite{Oe:affine} this was considered as modifying the action, here we consider it as giving rise to a Weyl observable $F\defeq \exp(\im\, D)$. Of course, we know from Section~\ref{sec:feynweyl} that both points of view are intimately related. Indeed, the following identity is essentially equivalent to equation (119) in \cite{Oe:affine}.

\begin{prop}
\label{prop:sfowfp}
Let $M$ be a region, $D\in\cobsl_M$, $F\defeq \exp(\im\, D)$, and $\tau\in L_{\partial M}$. Define $\hat{\tau}\in L_{\partial M}^\C$ as $\hat{\tau}\defeq \tau^{\mathrm{R}}-\im\tau^{\mathrm{I}}$, where $\tau=\tau^{\mathrm{R}}+J_{\partial M}\tau^{\mathrm{I}}$ and $\tau^{\mathrm{R}},\tau^{\mathrm{I}}\in L_{\tilde M}$. Then,
\begin{equation}
 \rho_M^F\left(\cohn_{\tau}\right)=\rho_M(\cohn_{\tau}) F\left(\hat{\tau}\right) \rho_M^F\left(\cohn_0\right) .
\label{eq:sfofactid}
\end{equation}
Moreover, we have
\begin{equation}
\rho_M^F\left(\cohn_0\right)=\exp\left(\frac{\im}{2}D(\eta_D)-\frac{1}{2}g_{\partial M}(\eta_D,\eta_D)\right) ,
\label{eq:sfovev}
\end{equation}
where $\eta_D\in A_{\tilde M}^D\cap J_{\partial M}L_{\partial M}$ is unique.
\end{prop}
\begin{proof}
We start with the second identity (\ref{eq:sfovev}) and the special element $\eta_D$. The existence and uniqueness of $\eta_D$ can be seen by taking any element $\phi\in A_{\tilde M}$ and decomposing it as $\phi=\phi^{\mathrm{R}}+J_{\partial M}\phi^{\mathrm{I}}$ with $\phi^{\mathrm{R}},\phi^{\mathrm{I}}\in L_{\partial M}$. Then, it is easy to see that $\eta_D=J_{\partial M}\phi^{\mathrm{I}}$ is unique. We apply now Lemma~\ref{lem:relomampl} with $\eta\defeq\eta_D$ and $\xi\defeq-\eta_D$. Combining this with the amplitude map (\ref{eq:amplcoh}) yields (\ref{eq:sfovev}).

We turn to the first identity (\ref{eq:sfofactid}). We set in Lemma~\ref{lem:relomampl} $\eta\defeq\eta_D$ and $\xi\defeq\tau-\eta_D$. Thus,
\begin{align}
\rho_M^F\left(\cohn_{\tau}\right) & =\rho_M\left(\cohn_{\tau-\eta_D}\right)
 \exp\left(\frac{\im}{2} D(\eta_D)+\im\,\omega_{\partial M}\left(\tau-\eta_D,\eta_D\right)\right)\\
 & =\rho_M\left(\cohn_{\tau}\right)
 \exp\left(\im\,\omega_{\partial M}\left(\tau^{\mathrm{R}},J_{\partial M}\tau^{\mathrm{I}}\right)
 +\frac{1}{2}g_{\partial M}\left(J_{\partial M}\tau^{\mathrm{I}},J_{\partial M}\tau^{\mathrm{I}}\right)\right. \nonumber\\
 &\qquad
 -\im\,\omega_{\partial M}\left(\tau^{\mathrm{R}},J_{\partial M}\tau^{\mathrm{I}}-\eta_D\right)
 -\frac{1}{2}g_{\partial M}\left(J_{\partial M}\tau^{\mathrm{I}}-\eta_D,J_{\partial M}\tau^{\mathrm{I}}-\eta_D\right) \nonumber\\
 &\qquad
 \left.
+\frac{\im}{2} D(\eta_D)+\im\,\omega_{\partial M}\left(\tau^{\mathrm{R}},\eta_D\right)\right)\\
 & =\rho_M\left(\cohn_{\tau}\right)
 \exp\left(2\im\,\omega_{\partial M}\left(\tau^{\mathrm{R}},\eta_D\right)
 +2\omega_{\partial M}\left(\tau^{\mathrm{I}},\eta_D\right)\right)
 \nonumber\\
 &\qquad
 \exp\left(\frac{\im}{2} D(\eta_D)-\frac{1}{2}g_{\partial M}(\eta_D,\eta_D)\right)
\end{align}
Here we have used expression (\ref{eq:amplcoh}) for the respective amplitudes in the first step. It remains to observe that with (\ref{eq:sfovev}) and (\ref{eq:lobprop}) from axiom (C5) we obtain expression (\ref{eq:sfofactid}).
\end{proof}

Let us emphasize that the three factors on the right hand side of (\ref{eq:sfofactid}) reproduce precisely the three factors on the right hand side of (\ref{eq:genfaampl}) in the example at hand. In particular, $\hat{\tau}$ in (\ref{eq:genfaampl}) is precisely given by the equation $\hat{\tau}= \tau^{\mathrm{R}}-\im\tau^{\mathrm{I}}$ induced by the decomposition $L_{\partial M}=L_{\tilde M}\oplus J_{\partial M} L_{\tilde M}$. The only factor for which the coincidence is not so obvious is the ``vacuum expectation value'' of the Weyl observable (\ref{eq:sfovev}). To see this, we rewrite (\ref{eq:sfovev}) as follows,
\begin{align}
\rho_M^F\left(\cohn_0\right) & =\exp\left(\frac{\im}{2}D(\eta_D)-\omega_{\partial M}\left(\eta_D,J_{\partial M}\eta_D\right)\right)\\
& = \exp\left(\frac{\im}{2}D\left(\eta_D-\im J_{\partial M}\eta_D\right)\right) .
 \label{eq:sfovevpol}
\end{align}
In the second line we have applied relation (\ref{eq:lobprop}) of axiom (C5) and taken the liberty of extending $D$ complex linearly to complexified configurations. Now, note that $J_{\partial M}\eta_D\in L_{\tilde{M}}$ since $\eta_D\in J_{\partial M} L_{\tilde{M}}$. That is $\eta_D-\im J_{\partial M}\eta_D$ is still a solution, now complex, of the inhomogeneous Klein-Gordon equation with source $\mu$. Moreover, on the boundary of $M$ we have $\eta_D-\im J_{\partial M}\eta_D\in P^+(L_{\partial M})$, where $P^+(L_{\partial M})\subseteq L_{\partial M}^{\C}$ is the polarized subspace of the complexified boundary solution space defined by
\begin{equation}
 P^+(L_{\partial M})\defeq \{\xi-\im J_{\partial M}\xi : \xi\in L_{\partial M}\} .
\end{equation}
This is precisely the Feynman boundary condition. That is, $\eta_D-\im J_{\partial M}\eta_D$ as a spacetime configuration takes the form,
\begin{equation}
 (\eta_D-\im J_{\partial M}\eta_D)(x)=\int G_F(x,x') \mu(x')\,\xd x' .
\end{equation}
Combining this with (\ref{eq:sfovevpol}) and (\ref{eq:srcobs}) yields the third factor in (\ref{eq:genfaampl}). 

It was shown in \cite{Oe:obsgbf} that a related identity holds for a much larger class of observables. This was termed the \emph{coherent factorization property} as it applies to coherent states. Concretely, it was suggested in \cite{Oe:obsgbf} that Schrödinger-Feynman quantization exhibits this property. On the other hand, for two other quantization schemes (Berezin-Toeplitz and normal ordering) it was rigorously shown to hold.

\begin{prop}[Coherent Factorization Property]
\label{prop:sfocfp}
Let $M$ be a regular region, $F\in\cobs_M$ and $\tau\in L_{\partial M}$. Define $\hat{\tau}\in L_{\partial M}^\C$ as $\hat{\tau}\defeq \tau^{\mathrm{R}}-\im\tau^{\mathrm{I}}$, where $\tau=\tau^{\mathrm{R}}+J_{\partial M}\tau^{\mathrm{I}}$ and $\tau^{\mathrm{R}},\tau^{\mathrm{I}}\in L_{\tilde M}$. Define $F'\in\cobs_M$ by
$F'(\xi)\defeq F(\xi+\hat{\tau})$. Then,
\begin{equation}
 \rho_M^F\left(\cohn_{\tau}\right)=\rho_M\left(\cohn_{\tau}\right) \rho_M^{F'}\left(\cohn_0\right) .
\label{eq:sfocfp}
\end{equation}
\end{prop}
\begin{proof}
Since $F\in\cobs_M$ there are $D_1,\dots,D_n\in \cobsl_M$ such that $F=\sum_{k=1}^n \lambda_k F_k$ with $F_k=\exp(\im\, D_k)$. Moreover, we have
\begin{equation}
F'=\sum_{k=1}^n \lambda_k F_k'=\sum_{k=1}^n\lambda_k F_k\left(\hat{\tau}\right) F_k
\end{equation}
since the $F_k$ are Weyl observables. Now by linearity of the quantization map $\quant_M$ together with identity (\ref{eq:sfofactid}) we have
\begin{multline}
\rho_M^{F}\left(\cohn_{\tau}\right)=\sum_{k=1}^n \lambda_k\,\rho_M^{F_k}\left(\cohn_{\tau}\right)=\rho_M\left(\cohn_{\tau}\right)\sum_{k=1}^n \lambda_k\, F_k\left(\hat{\tau}\right)\rho_M^{F_k}\left(\cohn_0\right)\\
=\rho_M\left(\cohn_{\tau}\right)\sum_{k=1}^n \lambda_k\, \rho_M^{F_k'}\left(\cohn_0\right)=\rho_M\left(\cohn_{\tau}\right)\rho_M^{F'}\left(\cohn_0\right) .
\end{multline}
\end{proof}

Note that the argument of Proposition~\ref{prop:sfocfp} is reversible. Using linearity of the quantization map $\quant_M$ we can deduce (\ref{eq:sfofactid}) from (\ref{eq:sfocfp}) for Weyl observables.

\subsection{More general observables}
\label{sec:genobs}

Weyl observables have another advantage besides that exploited in Section~\ref{sec:feynweyl}: They may serve as generators for other types of observables. Indeed, a linear observable $D\in\cobsl_M$ may be obtained as the first derivative of a corresponding Weyl observable with a parameter $\lambda\in\R$ inserted. Concretely,
\begin{equation}
 D=\left. -\im\frac{\partial}{\partial \lambda} \exp\left(\im \lambda D\right) \right|_{\lambda=0}
\end{equation}
Since the quantization map $\quant_M$ is linear we may commute it with the derivative. Defining $F\defeq \exp\left(\im \lambda D\right)$ we thus have,
\begin{equation}
 \rho_M^D=\left. -\im\frac{\partial}{\partial \lambda} \rho_M^F \right|_{\lambda=0} .
\end{equation}
In fact it is easy to evaluate this using Proposition~\ref{prop:sfowfp}. Thus, equations (\ref{eq:sfofactid}) and (\ref{eq:sfovev}) yield,
\begin{align}
 \rho_M^D\left(\cohn_\tau\right) & =\left. -\im\frac{\partial}{\partial \lambda} \rho_M(\cohn_{\tau}) \exp\left(\im\lambda D(\hat{\tau})+\frac{\im}{2}\lambda^2 D(\eta_D)-\frac{1}{2}\lambda^2 g_{\partial M}(\eta_D,\eta_D)\right)\right|_{\lambda=0}\\
& = \rho_M(\cohn_{\tau}) D(\hat{\tau}) .
\label{eq:qlinobscoh}
\end{align}
Not unexpectedly, the quantization of a linear observable is particularly simple.

Polynomial observables can be generated similarly. Thus, suppose we are interested in the product observable $D_1\cdots D_n$, where $D_1,\dots,D_n\in\cobsl_M$ are linear observables. Then,
\begin{equation}
 D_1\cdots D_n=\left. (-\im)^{n}\frac{\partial}{\partial \lambda_1}\cdots\frac{\partial}{\partial \lambda_n} \exp\left(\im \sum_{k=1}^n\lambda_k D_k\right) \right|_{\lambda_1=0,\dots,\lambda_n=0} .
\end{equation}
Setting
\begin{equation}
G\defeq D_1\cdots D_n\quad\text{and}\quad F\defeq \exp\left(\im \sum_{k=1}^n\lambda_k D_k\right) ,
\end{equation}
linearity of the quantization map $\quant_M$ may be used again, yielding
\begin{equation}
\rho_M^G=\left. (-\im)^{n}\frac{\partial}{\partial \lambda_1}\cdots\frac{\partial}{\partial \lambda_n} \rho_M^F\right|_{\lambda_1=0,\dots,\lambda_n=0} .
\end{equation}
To evaluate this, one may again use Proposition~\ref{prop:sfowfp}, although the result will be more complicated due to the contribution from the ``vacuum expectation value'' (\ref{eq:sfovev}) which becomes non-trivial. We merely exhibit the quadratic case $n=2$ here,
\begin{multline}
\rho_M^{D_1 D_2}\left(\cohn_\tau\right)=\rho_M\left(\cohn_\tau\right) \\
\left(D_1(\hat{\tau})D_2(\hat{\tau})-\frac{\im}{2} D_1(\eta_{D_2})-\frac{\im}{2} D_2(\eta_{D_1})
+g_{\partial M}\left(\eta_{D_1},\eta_{D_2}\right)\right) .
\label{eq:quadobs}
\end{multline}
We may also straightforwardly combine Weyl observables with polynomials.

It is easily seen that all the axioms are compatible with an extension of the observable algebra $\cobs_M$ from linear combinations of Weyl observables to, say, all possible products of polynomials with Weyl observables. This comes really all down to the linearity of the quantization map. Worries about the well definedness of observable maps do not arise since, as we have seen, evaluations on coherent states are always well defined. Thus, even though we have confined ourselves thus far mainly to Weyl observables, the principal results (axioms of Sections~\ref{sec:cobsax} and \ref{sec:compcor} and Proposition~\ref{prop:sfocfp}) hold for a much wider class of observables.

As a basic example of the identity (\ref{eq:quadobs}) we consider the relation between the vacuum expectation value of the time-ordered two-point function and the Feynman propagator in quantum field theory. To this end, we specialize the identity (\ref{eq:quadobs}) to the vacuum state and rewrite it, using (\ref{eq:lobprop}) of axiom (C5),
\begin{align}
& \rho_M^{D_1 D_2}\left(\cohn_0\right) = 
-\frac{\im}{2} D_1(\eta_{D_2})-\frac{\im}{2} D_2(\eta_{D_1})
+g_{\partial M}\left(\eta_{D_1},\eta_{D_2}\right) \\
& =
-\frac{\im}{2} D_1(\eta_{D_2})-\frac{\im}{2} D_2(\eta_{D_1})
+\omega_{\partial M}\left(\eta_{D_1},J_{\partial M}\eta_{D_2}\right)
+\omega_{\partial M}\left(\eta_{D_2},J_{\partial M}\eta_{D_1}\right) \\
& =
-\frac{\im}{2} D_1\left(\eta_{D_2}-\im J_{\partial M}\eta_{D_2}\right)
 -\frac{\im}{2} D_2\left(\eta_{D_1}-\im  J_{\partial M}\eta_{D_1}\right) .
\label{eq:quadobsfeyn}
\end{align}
Here we have extended $D_1$ and $D_2$ to complex linear observables on the complexified configuration space.

As in the example of Section~\ref{sec:factorization} we consider a real Klein-Gordon field in Minkowski space and take a region $M=[t,t']\times \R^3$ determined by a time interval $[t,t']$. The observables of interest are now the point-wise evaluations,
\begin{equation}
 D_1(\phi)=\phi(x_1),\qquad D_2(\phi)=\phi(x_2) ,
\end{equation}
where $x_1,x_2\in [t,t']\times \R^3$. Even though these observables are not continuous and thus not elements of $\cobsl_M$, we may treat them in a distributional sense, with certain limitations. The associated inhomogeneous Klein-Gordon equations have $\delta$-function sources at $x_1$ and $x_2$ respectively. Thus, $\eta_{D_1}-\im  J_{\partial M}\eta_{D_1}$ and $\eta_{D_2}-\im J_{\partial M}\eta_{D_2}$ are corresponding fundamental solutions. Moreover, they satisfy precisely the Feynman boundary conditions. (Compare the discussion in Section~\ref{sec:factorization}.) Thus, they are versions of the Feynman propagator,
\begin{equation}
(\eta_{D_1}-\im  J_{\partial M}\eta_{D_1})(x)=G_F(x,x_1),\quad
(\eta_{D_2}-\im  J_{\partial M}\eta_{D_2})(x)=G_F(x,x_2) .
\end{equation}
Inserting this into (\ref{eq:quadobsfeyn}) and taking into account the symmetry of the Feynman propagator yields the expected result,
\begin{equation}
\langle \cohn_0, \tord \phi(x_1)\phi(x_2) \cohn_0\rangle =
\rho_M^{D_1 D_2}(\cohn_0)=-\im\, G_F(x_1,x_2) .
\end{equation}

\subsection{Canonical commutation relations}
\label{sec:diracquant}

So far we have not mentioned the Dirac quantization condition (\ref{eq:diraccond}) of Section~\ref{sec:obsqft}. The main reason for this is that we have focused so far on Weyl observables. In quantum field theory these do not satisfy the Dirac quantization condition. With the considerations of Section~\ref{sec:genobs} we also have at our disposal linear observables, for which the Dirac quantization condition reduces to the \emph{canonical commutation relations}. These are satisfied in quantum field theory. We shall show in this section that they can be made sense of also in the present setting and are indeed satisfied for linear observables.

An immediate observation about the relation (\ref{eq:diraccond}) is that from the GBF perspective it is a statement for ``infinitesimally thin'' regions. This is so because the relation pretends to compose a region (where $D$ lives) with a copy of itself (where $E$ lives) to obtain the same region again (where the observable on the right hand side lives). We recall from Section~\ref{sec:geomax} that these regions are called \emph{slice regions}. Consider the slice region $\hat{\Sigma}$ associated to the hypersurface $\Sigma$ with boundary $\partial \hat{\Sigma}=\Sigma\cup\overline{\Sigma'}$, where $\Sigma'$ is a copy of $\Sigma$. For $D,E\in\cobsl_{\hat{\Sigma}}$ the relation (\ref{eq:diraccond}) can be transcribed as,\footnote{Note that we are committing slight abuse of notation here. The term $\rho_{\hat{\Sigma}}^E\aglue_{\Sigma} \rho_{\hat{\Sigma}}^D$ should strictly speaking be written as $\aglue_{\Sigma}\left(\rho_{\hat{\Sigma}}^E\aglue \rho_{\hat{\Sigma}}^D\right)$ etc.}
\begin{equation}
\rho_{\hat{\Sigma}}^E\aglue_{\Sigma} \rho_{\hat{\Sigma}}^D - \rho_{\hat{\Sigma}}^D\aglue_{\Sigma} \rho_{\hat{\Sigma}}^E=\rho_{\hat{\Sigma}}^K .
\label{eq:diracgbf}
\end{equation}
Here $K$ is given as
\begin{equation}
K=-\im\, (E,D),
\end{equation}
where the bracket is supposed to be the Poisson bracket. This seems to raise a difficulty. The observables here are not functions on the phase space $L_{\hat{\Sigma}}$, where the Poisson bracket is defined, but on the larger space $K_{\hat{\Sigma}}$. Simply restricting the observables $E,D$ to the subspace $L_{\hat{\Sigma}}\subseteq K_{\hat{\Sigma}}$ makes their Poisson bracket $(E,D)$ well defined. However, $(E,D)$ is then a function on $L_{\hat{\Sigma}}$ and not on $K_{\hat{\Sigma}}$. But $E$ and $F$ are linear, so their Poisson bracket is a constant function on $L_{\hat{\Sigma}}$ which we extend to a function with the same constant value on $K_{\hat{\Sigma}}$. So, with these details understood, the canonical commutation relations do make sense.

We take the structure of $L_{\hat{\Sigma}}$ as a symplectic vector space to come from identifying $L_{\hat{\Sigma}}$ with $L_{\Sigma}$ (rather than with $L_{\overline{\Sigma'}}$). For the orderings in (\ref{eq:diracgbf}) to correspond to the orderings (\ref{eq:diraccond}) this means that in $\rho_{\hat{\Sigma}}^E\aglue_{\Sigma} \rho_{\hat{\Sigma}}^D$ the $\Sigma$-side of $\hat{\Sigma}$ in the first term has to be glued to the $\Sigma'$-side of $\hat{\Sigma}$ in the second term.
For a linear observable $D$, the Hamiltonian vector field $x_D$ is translation invariant and may thus be identified with an element in $L_{\hat{\Sigma}}$. Concretely, $x_D$ is defined by the equation,
\begin{equation}
 D(\xi)=2\omega_{\Sigma}(\xi,x_D) \qquad\forall \xi\in L_{\hat{\Sigma}} .
\label{eq:hamlo}
\end{equation}
This equation is reminiscent of equation (\ref{eq:lobprop}) in axiom (C5). Indeed, writing the latter using $\eta_D$ as defined in Proposition~\ref{prop:sfowfp} we obtain a simple relation between $x_D\in L_\Sigma$ and $\eta_D\in L_{\partial\hat{\Sigma}}=L_\Sigma\times L_{\overline{\Sigma'}}$:
\begin{equation}
 \eta_D=\left(\frac{1}{2}x_D,-\frac{1}{2} x_D\right) .
\label{eq:hvectrel}
\end{equation}
To see this relation we first note that while $L_{\tilde{\hat{\Sigma}}}\subseteq L_{\partial\hat{\Sigma}}$ consists of the elements $(\xi,\xi)$, the space $J_{\partial\hat{\Sigma}}L_{\tilde{\hat{\Sigma}}}\subseteq L_{\partial\hat{\Sigma}}$ consists precisely of the elements $(\xi,-\xi)$ for $\xi\in L_{\Sigma}$. Thus, both sides of (\ref{eq:hvectrel}) are elements of $J_{\partial\hat{\Sigma}}L_{\tilde{\hat{\Sigma}}}\subseteq L_{\partial\hat{\Sigma}}$. It remains to verify that (\ref{eq:hamlo}) and (\ref{eq:lobprop}) are equivalent, given (\ref{eq:hvectrel}),
\begin{equation}
 2\omega_{\Sigma}(\xi,x_D)=2\omega_{\Sigma}\left(\xi,\frac{1}{2} x_D\right)
+ 2\omega_{\overline{\Sigma'}}\left(\xi,-\frac{1}{2} x_D\right)
=2\omega_{\partial\hat{\Sigma}}(\xi,\eta_D) .
\label{eq:relhamsol}
\end{equation}

Equipped with the Hamiltonian vector fields the Poisson bracket can be conveniently expressed in terms of the symplectic structure,
\begin{equation}
(E,D)=2\omega_{\Sigma}\left(x_E,x_D\right) .
\end{equation}
So we can write the canonical commutation relations (\ref{eq:diracgbf}) as follows,
\begin{equation}
\rho_{\hat{\Sigma}}^E\aglue_{\Sigma} \rho_{\hat{\Sigma}}^D - \rho_{\hat{\Sigma}}^D\aglue_{\Sigma} \rho_{\hat{\Sigma}}^E=-2\im\,\omega_{\Sigma}\left(x_E,x_D\right)\rho_{\hat{\Sigma}} .
\label{eq:diracgbf2}
\end{equation}
In order to prove that they hold it will be convenient to explore the operator product of observables first for Weyl observables.

\begin{prop}
\label{prop:woprod}
Let $\Sigma$ be a hypersurface and $\hat{\Sigma}$ the associated slice region. Let $D,E\in\cobsl_{\hat{\Sigma}}$ and $x_D,x_E\in L_{\Sigma}$ the associated Hamiltonian vector fields. Define $F\defeq \exp(\im\, D)$, $G\defeq\exp\left(\im\, E\right)$ and the order of composition as specified previously. Then,
\begin{equation}
\rho_{\hat{\Sigma}}^F\aglue_{\Sigma} \rho_{\hat{\Sigma}}^G=
\exp\left(-\frac{\im}{2}D(\eta_E)-\frac{\im}{2}E(\eta_D)+\im\,\omega_{\Sigma}(x_D,x_E)\right)
\rho_{\hat{\Sigma}}^{F\cdot G} .
\label{eq:woprod}
\end{equation}
\end{prop}
\begin{proof}
As usual we demonstrate this on coherent states. We remark that $(x_E,0)\in A_{\hat{\Sigma}}^E$ and $(0,-x_D)\in A_{\hat{\Sigma}}^D$ and thus by linearity of $s_{\hat{\Sigma}}$, $(x_E,-x_D)\in A_{\hat{\Sigma}}^{D+E}$. Then, given $\tau_1,\tau_2\in L_{\Sigma}$ we have,
\begin{align}
& \left(\rho_{\hat{\Sigma}}^F\aglue_{\Sigma} \rho_{\hat{\Sigma}}^G\right)
\left(\cohn_{\tau_1+x_E}\tens\cohn_{\tau_2-x_D}\right) \label{eq:wop1} \\
& = \int_{\hat{L}_{\Sigma}} \rho_{\hat{\Sigma}}^G\left(\cohn_{\tau_1+x_E}\tens\cohn_\xi\right)
 \rho_{\hat{\Sigma}}^F\left(\cohn_{\xi}\tens\cohn_{\tau_2-x_D}\right)\exp\left(\frac{1}{2}g_{\Sigma}(\xi,\xi)\right)\xd\nu_\Sigma(\xi)  \label{eq:wop2} \\
& = \exp\left(\frac{\im}{2}E\left((x_E,0)\right)+\frac{\im}{2}D\left((0,-x_D)\right)\right. \nonumber\\
&\qquad +\im\,\omega_{\partial\hat{\Sigma}}\left((\tau_1,\xi),(x_E,0)\right)
+\im\,\omega_{\partial\hat{\Sigma}}\left((\xi,\tau_2),(0,-x_D)\right)\bigg) \nonumber\\
&\qquad \int_{\hat{L}_{\Sigma}} \rho_{\hat{\Sigma}}\left(\cohn_{\tau_1}\tens\cohn_\xi\right)
 \rho_{\hat{\Sigma}}\left(\cohn_{\xi}\tens\cohn_{\tau_2}\right)\exp\left(\frac{1}{2}g_{\Sigma}(\xi,\xi)\right)\xd\nu_\Sigma(\xi)  \label{eq:wop3} \\
& = \exp\left(-\frac{\im}{2}D\left((x_E,0)\right)-\frac{\im}{2}E\left((0,-x_D)\right)\right)  \nonumber\\
&\qquad
 \exp\left(\frac{\im}{2}(D+E)\left((x_E,-x_D)\right)+\im\,\omega_{\partial\hat{\Sigma}}\left((\tau_1,\tau_2),(x_E,-x_D)\right)\right)
\rho_{\hat{\Sigma}}\left(\cohn_{\tau_1}\tens\cohn_{\tau_2}\right) \label{eq:wop4} \\
& = \exp\left(-\frac{\im}{2}D\left((x_E,0)\right)-\frac{\im}{2}E\left((0,-x_D)\right)\right) \rho_{\hat{\Sigma}}^{F\cdot G}\left(\cohn_{\tau_1+x_E}\tens\cohn_{\tau_2-x_D}\right)  \label{eq:wop5} \\
& =  \exp\left(-\frac{\im}{2}D\left(\left(\frac{1}{2}x_E,-\frac{1}{2}x_E\right)\right)-\frac{\im}{2}E\left(\left(\frac{1}{2} x_D,-\frac{1}{2}x_D\right)\right)\right. \nonumber\\
& \qquad \left.
-\frac{\im}{4}D\left(\left(x_E,x_E\right)\right)+\frac{\im}{4}E\left(\left(x_D, x_D\right)\right)\right)
 \rho_{\hat{\Sigma}}^{F\cdot G}\left(\cohn_{\tau_1+x_E}\tens\cohn_{\tau_2-x_D}\right) \label{eq:wop6} \\
& =  \exp\left(-\frac{\im}{2}D\left(\eta_E\right)-\frac{\im}{2}E\left(\eta_D\right)\right. \nonumber\\
& \qquad \left.
-\frac{\im}{2}\omega_{\Sigma}\left(x_E,x_D\right)+\frac{\im}{2}\omega_{\Sigma}\left(x_D, x_E\right)\right)
 \rho_{\hat{\Sigma}}^{F\cdot G}\left(\cohn_{\tau_1+x_E}\tens\cohn_{\tau_2-x_D}\right)  \label{eq:wop7}
\end{align}
The first step from (\ref{eq:wop1}) to (\ref{eq:wop2}) is given by the definition of the gluing operation for observable map. Here this is really the combination of two operations. The first is the gluing of two copies of $\hat{\Sigma}$ to their disjoint union (defined as in axiom (O2a)), the second is the gluing of this region to itself (defined as in axiom (O2b)) to obtain a single copy of $\hat{\Sigma}$. Note that the gluing anomaly factor in this case is unity as follows from core axiom (T3x). (\ref{eq:wop3}) is obtained by applying Lemma~\ref{lem:relomampl} to both observable maps. To obtain (\ref{eq:wop4}) core axiom (T5a) combined with (T5b) is applied. At the same time the argument of the exponential factor is manipulated in a straightforward way. Then, Lemma~\ref{lem:relomampl} is applied again to yield expression (\ref{eq:wop5}). In the following steps relation (\ref{eq:relhamsol}) is used as well as equation (\ref{eq:hamlo}).
\end{proof}

Before proceeding, let us remark on the possible impression of a tension between Proposition~\ref{prop:woprod} and the Dirac quantization condition on the one hand and composition correspondence on the other. Indeed, it might appear that composition correspondence implies that a modified version of equation (\ref{eq:woprod}) should hold, where the exponential factor is simply removed. Actually, composition correspondence does imply a similar relation, which we may write as,
\begin{equation}
\rho_{\hat{\Sigma}}^F\aglue_{\Sigma} \rho_{\hat{\Sigma}}^G=\rho_{\hat{\Sigma}}^{\tilde{F}\cdot \tilde{G}} .
\end{equation}
This follows from axioms (X2a) and (X2b). The definition of $\tilde{F}$ is given by,
\begin{equation}
 \tilde{F}=l_{(\hat{\Sigma}\cup\hat{\Sigma};\Sigma,\overline{\Sigma'})}\circ l_{(\hat{\Sigma};\hat{\Sigma}\cup\hat{\Sigma})}(F) .
\end{equation}
(The discussion of $\tilde{G}$ is analogous.) Here we have used the notation defined in axioms (CO2a) and (CO2b) of Section~\ref{sec:cobsax}. To explain, $l_{(\hat{\Sigma};\hat{\Sigma}\cup\hat{\Sigma})}$ extends $F$ from $K_{\hat{\Sigma}}$ to $K_{\hat{\Sigma}}\times K_{\hat{\Sigma}}$ trivially, i.e., by ignoring the second factor. Then, $l_{(\hat{\Sigma}\cup\hat{\Sigma};\Sigma,\overline{\Sigma'})}$ restricts to a dependence on the subspace $K_{\hat{\Sigma}}\subseteq K_{\hat{\Sigma}}\times K_{\hat{\Sigma}}$, obtained as the subspace of those configurations that match at the boundaries to be glued. This subspace is determined by the inclusion map $k_{\hat{\Sigma}\cup\hat{\Sigma};\Sigma,\overline{\Sigma'}}:K_{\hat{\Sigma}}\toi K_{\hat{\Sigma}}\times K_{\hat{\Sigma}}$ given in axiom (C7). Here, the fact that our observables depend on more general configurations and not merely on classical solutions is crucial. For, restricting to the latter, the map $k_{\hat{\Sigma}\cup\hat{\Sigma};\Sigma,\overline{\Sigma'}}$ takes the simple form
\begin{equation}
k_{\hat{\Sigma}\cup\hat{\Sigma};\Sigma,\overline{\Sigma'}}(\xi)=(\xi,\xi)\qquad \forall\xi\in L_{\hat{\Sigma}} .
\label{eq:keel}
\end{equation}
This is due to the second of the exact sequences (\ref{eq:c7d1}) in axiom (C7). As a consequence,
\begin{equation}
\tilde{F}(\xi)=F(\xi)\qquad \forall\xi\in L_{\hat{\Sigma}} .
\end{equation}
That is, restricted to classical solutions $F$ and $\tilde{F}$ are identical. This does not mean, however, that they are identical on all configurations nor that their quantizations coincide (as is made clear by the presence of the factor (\ref{eq:sfovev}) in Proposition~\ref{prop:sfowfp}). Indeed, take an element $\xi\in K_{\hat{\Sigma}}$ that has a boundary image not in the subspace $L_{\tilde{\hat{\Sigma}}}$. Then, $k_{\hat{\Sigma}\cup\hat{\Sigma};\Sigma,\overline{\Sigma'}}$ cannot take the form (\ref{eq:keel}) for this element, as that would violate the first of the exact sequences (\ref{eq:c7d1}) in axiom (C7). Thus, $\tilde{F}$ and $F$ are in general not identical. Indeed, Proposition~\ref{prop:woprod} precisely quantifies part of this difference.

An immediate consequence of Proposition~\ref{prop:woprod} are the Weyl relations,\footnote{Note that the Weyl relations are formulated here for quantizations of exponentials of imaginary linear observables rather than for exponentials of imaginary quantized linear observables as is customary.}
\begin{equation}
\rho_{\hat{\Sigma}}^F\aglue_{\Sigma} \rho_{\hat{\Sigma}}^G=
\rho_{\hat{\Sigma}}^G\aglue_{\Sigma} \rho_{\hat{\Sigma}}^F\exp\left(2\im\,\omega_{\Sigma}(x_D,x_E)\right) .
\end{equation}
It is also straightforward now, using the derivative method of Section~\ref{sec:genobs}, to extract the corresponding result for linear observables from Proposition~\ref{prop:woprod}.
\begin{prop}
\label{prop:loprod}
Let $\Sigma$ be a hypersurface and $\hat{\Sigma}$ the associated slice region. Let $D,E\in\cobsl_{\hat{\Sigma}}$ and $x_D,x_E\in L_{\Sigma}$ be the associated Hamiltonian vector fields. Then,
\begin{equation}
\rho_{\hat{\Sigma}}^D\aglue_{\Sigma} \rho_{\hat{\Sigma}}^E=
\rho_{\hat{\Sigma}}^{D\cdot E}+\left(\frac{\im}{2}D(\eta_E)+\frac{\im}{2}E(\eta_D)-\im\,\omega_{\Sigma}(x_D,x_E)\right)\rho_{\hat{\Sigma}} .
\label{eq:loprod}
\end{equation}
\end{prop}
\begin{proof}
We replace in expression (\ref{eq:woprod}) of Proposition~\ref{prop:woprod}, $D$ with $\lambda D$ and $E$ with $\mu E$. Then we take on both sides the derivative
\begin{equation}
-\frac{\partial}{\partial \lambda}\frac{\partial}{\partial \mu}
\end{equation}
and evaluate at $\lambda=\mu=0$.
\end{proof}
Finally, from expression (\ref{eq:loprod}) it is immediate to derive the canonical commutation relations in the form (\ref{eq:diracgbf2}).

\subsection{Comparison with other quantization schemes}
\label{sec:compquant}

In this section we shall compare the Schrödinger-Feynman quantization scheme for observables with two other natural quantization schemes in the same context of linear field theory. These other quantization schemes are based on viewing observables as functions on phase space rather than on spacetime configuration space. Thus, as explained in Section~\ref{sec:obsqft}, composition correspondence must fail for these quantization schemes. We shall be able to quantify this failure.

The first quantization scheme to be considered is \emph{Berezin-Toeplitz quantization}. In the context of the GBF this was introduced in \cite{Oe:obsgbf}. It takes an extremely simple form in the holomorphic representation. Let $M$ be a region and $F:L_{\tilde M}\to\C$ be a function (with suitable analyticity properties, see \cite{Oe:obsgbf}). Then its Berezin-Toeplitz quantization, denoted here $\rho_M^{\ano{F}}:\cH_{\partial M}\to\C$, is given for a state $\psi\in\cH_{\partial M}$ by the integral formula,
\begin{equation}
 \rho_M^{\ano{F}}(\psi)\defeq\int_{\hat{L}_{\tilde M}} \psi^\hr(\xi) F(\xi)\,\xd\nu_{\tilde M}(\xi) .
\end{equation}
Here $\psi^\hr$ denotes the wave function of the state $\psi$ in the holomorphic representation. As shown in \cite{Oe:obsgbf}, this quantization scheme reduces in the case of a slice region precisely to anti-normal ordering. It was also shown in \cite{Oe:obsgbf} that this quantization scheme satisfies the coherent factorization property, i.e., the analogue of Proposition~\ref{prop:sfocfp}. Restricting to Weyl observables yields the following analogue of the factorization property of Proposition~\ref{prop:sfowfp}.

\begin{prop}
\label{prop:btwfp}
Let $M$ be a region, $D:L_{M}\to\R$ linear and continuous and $\tau\in L_{\partial M}$. Define $F:L_{M}\to\C$ by $F\defeq \exp(\im\, D)$. Define $\hat{\tau}\in L_{\partial M}^\C$ as $\hat{\tau}\defeq \tau^{\mathrm{R}}-\im\tau^{\mathrm{I}}$, where $\tau=\tau^{\mathrm{R}}+J_{\partial M}\tau^{\mathrm{I}}$ and $\tau^{\mathrm{R}},\tau^{\mathrm{I}}\in L_{\tilde M}$. Then,
\begin{equation}
 \rho_M^{\ano{F}}\left(\cohn_{\tau}\right)=\rho_M(\cohn_{\tau}) F\left(\hat{\tau}\right) \rho_M^{\ano{F}}\left(\cohn_0\right) .
\label{eq:btfactid}
\end{equation}
Moreover, we have
\begin{equation}
\rho_M^{\ano{F}}\left(\cohn_0\right)=\exp\left(-g_{\partial M}(\eta_D,\eta_D)\right) ,
\label{eq:btvev}
\end{equation}
where $\eta_D\in J_{\partial M}L_{\partial M}$ is uniquely determined by $D(\xi)=2\omega_{\partial M}(\xi,\eta_D)$ for all $\xi\in L_M$.
\end{prop}
\begin{proof}
The coherent factorization property (i.e., the analogue of relation (\ref{eq:sfocfp}) of Proposition~\ref{prop:sfocfp}) for the Berezin-Toeplitz quantization scheme was proven in Proposition~4.1 of \cite{Oe:obsgbf}. Linearity of the quantization scheme then implies relation (\ref{eq:btfactid}). We proceed to demonstrate relation (\ref{eq:btvev}).
\begin{align}
\rho_M^{\ano{F}}\left(\cohn_0\right) & =\int_{\hat{L}_{\tilde M}} F(\xi)\,\xd\nu_{\tilde{M}}(\xi)\\
& =\int_{\hat{L}_{\tilde M}} \exp\left(2\im\, \omega_{\partial M}(\xi,\eta_D)\right)\,\xd\nu_{\tilde{M}}(\xi)\\
& =\int_{\hat{L}_{\tilde M}} \exp\left(-\im\, g_{\partial M}(\xi,J_{\partial M}\eta_D)\right)\,\xd\nu_{\tilde{M}}(\xi) \label{eq:bt3} \\
& =\exp\left(-g_{\partial M}(J_{\partial M}\eta_D,J_{\partial M}\eta_D)\right) \label{eq:bt4} \\
& =\exp\left(-g_{\partial M}(\eta_D,\eta_D)\right)
\end{align}
The notation of the integrals here is as in \cite{Oe:holomorphic}. The step from (\ref{eq:bt3}) to (\ref{eq:bt4}) is performed with techniques as used in that paper. Concretely, the $\im$ in the integrand is replaced with a complex variable, it is noted that the integrand is holomorphic in this variable. The integral is then performed for real values of the variable. Since the integral must also be holomorphic, the variable in the resulting expression is replaced again by $\im$. The integral itself is evaluated using Proposition~3.11 of \cite{Oe:holomorphic}.
\end{proof}

The other quantization scheme to be compared here is \emph{normal ordered quantization}. This was adapted to the GBF in \cite{Oe:obsgbf}. It takes a particularly simple form for coherent states. Thus, given a region $M$, a function $F:L_{\tilde M}\to\C$ (again with sufficient analyticity properties) we denote the quantization of $F$ by $\rho_M^{\no{F}}:\cH_{\partial M}\to\C$. Given $\tau\in L_{\partial M}$ we then have by definition,
\begin{equation}
\rho_M^{\no{F}}\left(\cohn_{\tau}\right)\defeq \rho_M\left(\cohn_{\tau}\right) F(\hat{\tau}) .
\end{equation}
In the special case of a slice region this yields precisely the usual concept of normal ordering. Also this quantization scheme satisfies the coherent factorization property. Moreover, it is immediate to see that for a Weyl observable $F$ we have the ``vacuum expectation value'',
\begin{equation}
\rho_M^{\no{F}}\left(\cohn_{0}\right)=1 .
\label{eq:novev}
\end{equation}

Unsurprisingly, all three quantization schemes coincide for linear observables. Indeed, on coherent states the quantum observable map is in all cases given by the analogue of expression (\ref{eq:qlinobscoh}). Thus, in particular, all satisfy the canonical commutation relations for linear observables, (\ref{eq:diracgbf2}). For Weyl observables the difference between the schemes is given by a constant, depending only on the observable. This follows from the coherent factorization property. Moreover, considering a suitable class of observables, both, the Berezin-Toeplitz scheme and the normal ordered scheme are easily seen to satisfy axioms (O1) and (O2a) of Section~\ref{sec:obsgbf} as well as axioms (X1) and (X2a) of Section~\ref{sec:compcor}. We shall now consider axioms (O2b) and (X2b) (composition correspondence). To this end we fix the observable algebras to consist of linear combinations of Weyl observables. More precisely, we shall assume, for each quantization scheme, an adapted version of the axioms of Section~\ref{sec:classlinobs}.

\begin{prop}
Let $M$ be a region with its boundary decomposing as a disjoint union $\partial M=\Sigma_1\cup\Sigma\cup\overline{\Sigma'}$ and $M_1$ as given in (T5b). Let $D\in \cobsl_M$ and $D_1\in \cobsl_{M_1}$ be the induced observable. Define $\eta_D\in J_{\partial M}L_{\tilde{M}}$ and $\eta_{D_1}\in J_{\partial M_1}L_{\tilde{M}_1}$ as previously. Let $F$, $F_1$ be the corresponding Weyl observables. Then, there exists $\eta_\Sigma\in L_\Sigma$ such that $(\eta_{D_1},\eta_\Sigma,\eta_\Sigma)-\eta_D\in L_{\tilde M}$. Moreover,
\begin{align}
& \aglue_{\Sigma}\left(\rho_M^{\ano{F}}\right)=\rho_{M_1}^{\ano{F_1}} \nonumber\\
& \quad \exp\left(\im\omega_{\partial M}\left((\eta_{D_1},\eta_{\Sigma},\eta_{\Sigma}),\eta_D\right)+\frac{1}{2} g_{\partial M_1}(\eta_{D_1},\eta_{D_1})-\frac{1}{2}g_{\partial M}(\eta_{D},\eta_{D})\right) ,
\label{eq:ccvano} \\
& \aglue_{\Sigma}\left(\rho_M^{\no{F}}\right)=\rho_{M_1}^{\no{F_1}} \nonumber\\
& \quad \exp\left(\im\omega_{\partial M}\left((\eta_{D_1},\eta_{\Sigma},\eta_{\Sigma}),\eta_D\right)-\frac{1}{2} g_{\partial M_1}(\eta_{D_1},\eta_{D_1})+\frac{1}{2}g_{\partial M}(\eta_{D},\eta_{D})\right) .
\label{eq:ccvno}
\end{align}
\end{prop}
\begin{proof}
Since the Schrödinger-Feynman quantization scheme satisfies axiom (X2b), that is the equation (\ref{eq:obsglueax}), we derive the corresponding equations for the other quantization schemes by relating them to this scheme through the coherent factorization property. In this manner of proceeding we make the implicit assumption that the observable $D$ extends to an observable on the configuration space $K_M$. However, we shall see that this extension is irrelevant.

As for the existence of $\eta_\Sigma$ observe that it needs to satisfy the condition
\begin{equation}
D(\xi)=2\omega_{\partial M}(\xi,(\eta_{D_1},\eta_\Sigma,\eta_\Sigma))
\label{eq:ccv1}
\end{equation}
for all $\xi\in L_{\tilde{M}}$. This is automatic for those $\xi$ taking the form $(\xi_1,\xi_\Sigma,\xi_\Sigma)$ as these are in $L_{\tilde{M}_1}$. It is thus sufficient to satisfy this for $\xi$ taking the form $(\xi_1,\xi_\Sigma,-\xi_\Sigma)$. Then equation (\ref{eq:ccv1}) takes the form
\begin{equation}
D(\xi)=2\omega_{\Sigma_1}(\xi_1,\eta_{D_1})
 +4\omega_\Sigma(\xi_\Sigma,\eta_\Sigma) .
\label{eq:ccv2}
\end{equation}
We note that if $\xi=(\xi_1,\xi_\Sigma,-\xi_\Sigma)$ and $\xi'=(\xi_1',\xi_\Sigma,-\xi_\Sigma)$ are both in $L_{\tilde{M}}$ we would have $\xi_1-\xi_1'\in L_{\tilde{M}_1}$ and hence $D(\xi)=D(\xi')$ in the above expression. Hence, there exists $\eta_\Sigma$ such that (\ref{eq:ccv2}) is satisfied for all $\xi=(\xi_1,\xi_\Sigma,-\xi_\Sigma)\in L_{\tilde{M}}$ and thus (\ref{eq:ccv1}) for all $\xi\in  L_{\tilde{M}}$. We proceed to observe that in the Schrödinger-Feynman setting $(\eta_{D_1},\eta_\Sigma,\eta_\Sigma)\in A_{\tilde{M}}^D$.

We start by demonstrating (\ref{eq:ccvno}). Using the coherent factorization property we obtain,
\begin{equation}
\aglue_{\Sigma}\left(\rho_M^{\no{F}}\right)=
\aglue_{\Sigma}\left(\rho_M^{F}\right)
\frac{\rho_M^{\no{F}}(\cohn_0)}{\rho_M^{F}(\cohn_0)}
= \rho_{M_1}^{F_1}
\frac{\rho_M^{\no{F}}(\cohn_0)}{\rho_M^{F}(\cohn_0)}
=
\rho_{M_1}^{\no{F_1}} \frac{\rho_M^{\no{F}}(\cohn_0)\rho_{M_1}^{F_1}(\cohn_0)}{\rho_{M_1}^{\no{F_1}}(\cohn_0)\rho_M^{F}(\cohn_0)} .
\end{equation}
Noting (\ref{eq:novev}) it remains to evaluate the quotient,
\begin{align}
& \frac{\rho_{M_1}^{F_1}(\cohn_0)}{\rho_M^{F}(\cohn_0)}
  = \exp\left(\frac{\im}{2}D_1(\eta_{D_1})-\frac{\im}{2}D(\eta_{D})
 -\frac{1}{2}g_{\partial M_1}(\eta_{D_1},\eta_{D_1})
 +\frac{1}{2}g_{\partial M}(\eta_{D},\eta_{D})\right) \\
 & = \exp\left(\frac{\im}{2}D((\eta_{D_1},\eta_\Sigma,\eta_\Sigma)-\eta_{D})
 -\frac{1}{2}g_{\partial M_1}(\eta_{D_1},\eta_{D_1})
 +\frac{1}{2}g_{\partial M}(\eta_{D},\eta_{D})\right) \\
& = \exp\left(\im\omega_{\partial M}\left((\eta_{D_1},\eta_{\Sigma},\eta_{\Sigma})-\eta_D,\eta_D\right)-\frac{1}{2} g_{\partial M_1}(\eta_{D_1},\eta_{D_1})+\frac{1}{2}g_{\partial M}(\eta_{D},\eta_{D})\right) \\
& = \exp\left(\im\omega_{\partial M}\left((\eta_{D_1},\eta_{\Sigma},\eta_{\Sigma}),\eta_D\right)-\frac{1}{2} g_{\partial M_1}(\eta_{D_1},\eta_{D_1})+\frac{1}{2}g_{\partial M}(\eta_{D},\eta_{D})\right) .
\end{align}
Showing (\ref{eq:ccvano}) is now straightforward by taking account of the difference between (\ref{eq:btvev}) and (\ref{eq:novev}). 
\end{proof}

While this result quantifies the violation of axiom (X2b) by Berezin-Toeplitz quantization and by normal ordered quantization, at the same time it shows that these schemes do satisfy axiom (O2b). Indeed, relations (\ref{eq:ccvano}) and (\ref{eq:ccvno}) show that the gluing operation of (T5b) applied to a Weyl observable in $\cobsw_M$ yields the multiple of a Weyl observable in $\cobsw_{M_1}$. Thus, taking the spaces $\cobs_M$ to consist of all linear combinations of elements of $\cobsw_M$, i.e.\ Weyl observables, yields the closure condition embodied by axiom (O2b). Moreover, this extends to more general observables generated as in Section~\ref{sec:genobs}.


\section{Conclusions and outlook}
\label{sec:conclusions}

The main result of Section~\ref{sec:sfaffine} of this paper is the constructive demonstration that the Feynman path integral, together with the Schrödinger representation, can be made into a rigorous and functorial quantization scheme for linear and affine field theory in the context of the general boundary formulation (GBF). Moreover, the so defined Schrödinger-Feynman quantization scheme is shown to be equivalent to the holomorphic quantization scheme introduced in \cite{Oe:holomorphic,Oe:affine} and based on geometric quantization. Even though we have not provided the details here it should be clear that this result may also serve as a rigorous underpinning for previous, more heuristic applications of Schrödinger-Feynman quantization in the GBF \cite{Oe:timelike,Oe:kgtl,CoOe:spsmatrix,CoOe:smatrixgbf,CoOe:2deucl,Col:desitterletter,CoOe:unitary,Col:desitterpaper,CoDo:smatrixcsp}.

In the spirit of the remarks made at the beginning of the introduction, one might expect the most useful applications of the presented quantization scheme to arise precisely in those areas where the conceptual basis of the standard approach to quantum field theory becomes shaky. One such area is quantum field theory in curved spacetime. Here we have in mind in particular settings where spacetime fails to be globally hyperbolic, such as Anti-de~Sitter space \cite{CDO:adsproc}.

In the second part of the paper, Section~\ref{sec:obs}, we have extended the Schrödinger-Feynman quantization scheme to include observables. Crucially, the relevant concept of quantum observable here is that of an observable map, introduced in \cite{Oe:obsgbf}, and not the more conventional one of an operator on Hilbert space (although the latter can be recovered from the former). As was argued already in \cite{Oe:obsgbf} and was made rigorous here, the former allows us to capture operationally relevant properties of observables in quantum field theory which the latter does not. These properties concern the composition of observables and are manifest in the time-ordered product and conveniently encoded through the Feynman path integral. In the present framework these properties are captured through what we have called here \emph{composition correspondence}, formalized in Section~\ref{sec:compcor}. We recall that this basically means that the composition of the quantization of classical observables with disjoint spacetime support equals the quantization of the ordinary product of the classical observables.

For the Feynman quantization of observables we focused first on Weyl observables (i.e., observables that are exponentials of imaginary linear observables). On the one hand these have the advantage that a much larger class of observables can be generated from them by functional differentiation (Section~\ref{sec:genobs}). On the other hand, they are simple enough so that we could derive in Section~\ref{sec:feynweyl} the closed formula (\ref{eq:wobscoh}) for their quantization. The crucial observation here was that from the point of view of the path integral, the insertion of a Weyl observable is essentially the same thing as modifying the action by adding a linear term. This allowed us to recur to a corresponding result of Section~\ref{sec:feynquant} from the first part of the paper. Based on these ingredients we were able in Section~\ref{sec:gbqftlow} to give a rigorous proof that the so defined quantization of classical linear field theory with observables defined in Section~\ref{sec:classlinobs} satisfies not only the core axioms of the GBF, but also the axioms concerning quantum observables and their relation to classical observables, including composition correspondence.

Another property of quantum field theory that generalizes nicely in the present setting is the form of the ``generating function'' or ``kernel'' of the S-matrix. We showed in Section~\ref{sec:factorization} that this is an example of a much more general factorization formula (Proposition~\ref{prop:sfowfp}). As in ordinary quantum field theory this suggests itself as the starting point of a perturbation theory. Incidentally, this means that even in (general boundary) quantum theories quite unlike ordinary quantum field theory, e.g., which are not based on a metric background spacetime, interactions can be expanded in terms of some kind of Feynman diagrams.

A version (Proposition~\ref{prop:sfocfp}) of the mentioned factorization formula which we call the \emph{coherent factorization property} was already introduced in \cite{Oe:obsgbf} and shown there to hold for two other quantization schemes: Berezin-Toeplitz quantization and normal ordering. The latter are ``simpler'' quantization schemes than the one presented here in the sense that observables there are only functions on phase space and not on spacetime configuration space. However, as we argued in Section~\ref{sec:obsqft} this implies that they cannot satisfy composition correspondence. Indeed, we were able to quantify this failure here in Section~\ref{sec:compquant}.

As explained in the paper, observables on a slice region, i.e., a region that is ``infinitesimally thin'' can be identified with operators on the associated Hilbert space. We were able to prove in Section~\ref{sec:diracquant} that linear observables of this type satisfy the canonical commutation relations in the presented quantization scheme. In the special case of quantum field theory with the slice region determined by an equal-time hypersurface these are precisely the ordinary equal-time commutation relations. This means in particular that we recover the ordinary operator product for observables in those ``infinitesimal'' regions.

If there is a well-defined notion of time-evolution, such as in Minkowski space, it is also possible to recover the operator product of general observables. Fix an equal-time hypersurface $\Sigma_1$. For simplicity suppose an observable $F$ of interest has support only in the future of the hypersurface. Say $\Sigma_2$ is an equal-time hypersurface in the future of the support of $F$. Both, the observable map associated to $F$ and the amplitude map for the so defined region can be converted to maps $\cH_{\Sigma_1}\to\cH_{\Sigma_2}$ between the boundary component Hilbert spaces. Call these $\check{\rho}^F$ and $\check{\rho}$ respectively. Then the composition $\check{\rho}^{-1}\circ \check{\rho}^F$ yields an operator on $\cH_{\Sigma_1}$. The operator product induced in this way is precisely the usual operator product of quantum field theory in the case of Minkowski space. Note that this method of composing forward and backward propagation in time is well known in standard quantum theory \cite{Sch:brownianosc}.

In any case, it is clear that the operator product is a less fundamental notion from the present perspective than the time-ordered product. Rather than a reason for worry this confirms that we are on the right track in extracting the (operationally relevant) features of observables in quantum field theory. One may take this as an indication to speculate that perhaps the canonical commutation relations are not such a natural quantization condition to impose after all. Perhaps, one should rather take the composition correspondence property more seriously. Of course, this in itself does not seem to be sufficient to single out ``good'' quantizations. As a first thought, it might be complemented by some kind of surjectivity and injectivity conditions for the quantization map $\quant_M$ from classical to quantum observable spaces, (with injectivity possibly restricted to the behavior of observables on classical solutions).

A slightly strange aspect of classical observables on a slice region $\hat{\Sigma}$ is that it is unclear what the associated spacetime configuration space $K_{\hat{\Sigma}}$ should be. It can definitely not be made equal to the phase space $L_{\hat{\Sigma}}$. One may circumvent this problem by simply not defining classical observables on slice regions. This would have no further effect on any of the axioms or results obtained, with the exception of the canonical commutation relations of Section~\ref{sec:diracquant}. However, these could then be recovered by a limit applied to the procedure outlined above for obtaining operators for regular regions. Even for regular regions we have made no attempt here to further specify the nature of the spacetime-configuration spaces $K_M$. However, as shown in the explicit examples in Sections~\ref{sec:factorization} and \ref{sec:genobs}, we may be able to clearly characterize observables of interest without actually specifying the relevant space $K_M$ explicitly. Another approach to this problem would be to consider $K_M$ as the union of all the spaces $A_M^D$ of solutions associated to all possible linear observables $D$ in the region $M$.

Let us also mention a weakness of the axiomatic system presented in Section~\ref{sec:classlinobs}. This has to do with the relation between interior and boundary data for observables. In most physically sensible situations one should expect a classical solution near the boundary to have at most one unique prolongation to all of the interior of the region.\footnote{Recall that we are in a setting without gauge symmetries here.} Indeed, it is conceivable that physically realistic theories can be formulated exclusively with regions that satisfy this requirement. Nevertheless, the axioms for linear or affine classical field theory presented here and in \cite{Oe:holomorphic,Oe:affine} allow for the possibility of a degeneracy, i.e., distinct classical solutions in the interior of a region may not be distinguishable near the boundary. One motivation for this degeneracy is to admit certain topological quantum field theories that are of great mathematical interest. The mentioned situation arises in particular when gluing regions with boundaries to a region without boundary, for example to obtain invariants of manifolds. Now axiom (C5) of Section~\ref{sec:classlinobs} requires that the value of observables on solutions in the interior of a region shall nevertheless only depend on the behavior of those solutions near the boundary. Indeed, the general validity of the composition correspondence axiom (X2b) (Section~\ref{sec:compcor}) hinges on that assumption. This may severely and unreasonably limit the set of admissible observables in theories where degeneracies are important. One way to remedy this would be to not require the relation (\ref{eq:lobprop}) of axiom (C5) for all observables in a region, but only for a specified subclass. The commutative diagram of axiom (X2b) would then only hold for those observables which are in this special class, both for the unglued as well as for the glued region. On the other hand, axiom (O2b) may still be valid more generally. (Compare the methods of Section~\ref{sec:compquant} for an indication how this might be the case.)

Finally, in case this has not become sufficiently clear, we emphasize again that all the obtained results apply to ordinary quantum field theory in particular (at least with admissible hypersurfaces defined to be equal-time hypersurfaces in Minkowski space). On the other hand, they also apply to a potentially much larger class of theories, including theories not based on a metric background spacetime.

\subsection*{Acknowledgments}

I would like to thank Daniele Colosi for many stimulating discussions before and during preparation of this paper. Moreover, I thank Daniele Colosi and Max Dohse for reading a draft version of this paper and alerting me to various typos.

\bibliographystyle{stdnodoi} 
\bibliography{stdrefsb}
\end{document}